\documentclass[twoside,11pt]{article}

\usepackage{jmlr2e}
\usepackage{lastpage}

\jmlrheading{26}{2025}{1-\pageref{LastPage}}{10/24; Revised 11/25}{11/25}{24-1662}{Juraj Bodik and Valérie Chavez-Demoulin}

\ShortHeadings{Identifiability of Causal Graphs under CPCM}{Bodik and Chavez-Demoulin}
\firstpageno{1}

\usepackage{etoolbox}
\usepackage{booktabs}
\usepackage[table,xcdraw]{xcolor}
\usepackage[usestackEOL]{stackengine}
\usepackage[T1]{fontenc}
\usepackage[utf8]{inputenc}
\usepackage{bm}
\usepackage{graphicx}
\usepackage{subcaption}
\usepackage{amsmath}
\usepackage{amsfonts}
\usepackage{mathtools}
\usepackage{xcolor}
\usepackage{float}
\usepackage{enumitem}
\usepackage{kantlipsum}
\usepackage{amssymb}
\usepackage{amsbsy}
\usepackage{bbm}
\usepackage{pifont}
\usepackage{caption}
\usepackage{multirow}
\usepackage[ruled,vlined]{algorithm2e}
\usepackage{listings}
\usepackage{dirtytalk}
\usepackage{chngcntr}
\usepackage{apptools}
\usepackage{mathrsfs}
\usepackage{wrapfig}
\usepackage{makecell}
\usepackage{tikz}
\usetikzlibrary{positioning,shapes,arrows}
\usepackage{tabularx}
\usepackage{siunitx}
\usepackage{xcolor,colortbl}
\usepackage{threeparttable}

\definecolor{HeaderGray}{gray}{0.98}
\definecolor{RowAlt}{gray}{0.97}
\definecolor{Accent}{HTML}{2D6CDF}

\DeclareMathOperator*{\argmin}{arg\,min}

\newtheorem{consequence}[theorem]{Consequence}

\newtheorem{innercustomthm}{Theorem}
\newenvironment{customthm}[1]
  {\renewcommand\theinnercustomthm{#1}\innercustomthm}
  {\endinnercustomthm}

\newtheorem{innercustomlem}{Lemma}
\newenvironment{customlem}[1]
  {\renewcommand\theinnercustomlem{#1}\innercustomlem}
  {\endinnercustomlem}

\newtheorem{innercustomprop}{Proposition}
\newenvironment{customprop}[1]
  {\renewcommand\theinnercustomprop{#1}\innercustomprop}
  {\endinnercustomprop}

\newtheorem{innercustomconsequence}{Consequence}
\newenvironment{customconsequence}[1]
  {\renewcommand\theinnercustomconsequence{#1}\innercustomconsequence}
  {\endinnercustomconsequence}

\newenvironment{hproof}{%
  \par\medskip\noindent{\bf Idea of the proof\ }%
}{\hfill\BlackBox\\[2mm]}

\newcommand{\indep}{\perp \!\!\! \perp}

\hypersetup{
    colorlinks,
    linkcolor={black},
    citecolor={blue!50!black},
    urlcolor={blue!80!black}
}

\graphicspath{{figures/}}

\begin{document}

% Required by JMLR: editor
\editor{Jin Tian}

\title{Identifiability of Causal Graphs under Non-Additive Conditionally Parametric Causal Models}

\author{\name Juraj Bodik \email juraj.bodik@unil.ch \\
       \addr HEC Lausanne\\
       University of Lausanne\\
       CH-1015 Lausanne, Switzerland
       \AND
       \name Valérie Chavez-Demoulin \email valerie.chavez@unil.ch \\
       \addr HEC Lausanne\\
       University of Lausanne\\
       CH-1015 Lausanne, Switzerland}

\maketitle

\begin{abstract}
Existing approaches to causal discovery often rely on restrictive modeling assumptions that limit their applicability in real-world settings, particularly when data are heavy-tailed or contain a mixture of discrete and continuous variables. Identifiability of causal graphs has been established under several structural models, including linear non-Gaussian models, post-nonlinear models, and location-scale models. However, these frameworks may not capture the diversity of distributions observed in practice. To address this, we introduce Conditionally Parametric Causal Models (CPCM), a flexible class of models where the conditional distribution of the effect, given its cause, belongs to a known parametric family such as Gaussian, Poisson, Gamma, or Pareto. These models are adaptable to a wide range of practical situations, where the cause influences not only the mean but also the variance or tail behavior of the effect. We demonstrate the identifiability of CPCM by leveraging the concept of sufficient statistics. Furthermore, we propose an algorithm for estimating the causal structure from random samples drawn from CPCM. We evaluate the empirical properties of our methodology on various datasets, demonstrating state-of-the-art performance across multiple benchmarks.
\end{abstract}


% ----------------------------------------------------------
% MAIN TEXT
% ----------------------------------------------------------
\pagenumbering{arabic}

\section{Introduction}
\label{introduction}

Understanding causal relations, as opposed to mere statistical associations, allows us to predict the effects of interventions that modify a system \citep{TheBookOfWhy}. Identifying such causal structures is central to many scientific disciplines \citep{Rubin}. Yet different data-generating mechanisms can lead to the same observational distribution, making causal inference inherently challenging. While observing a system under interventions can reliably reveal its causal structure, performing interventions is often expensive, ethically problematic \citep{epidem_application}, or simply unfeasible. This motivates the growing focus on estimating causal structure directly from observational data.

Over recent decades, extensive work has focused on building a rigorous mathematical framework for the “language’’ of causal inference, largely formalized through structural causal models (SCMs) \citep{Pearl_book}. Consider random variables $\mathbf{X} = (X_1, \dots, X_d)^\top \in \mathbb{R}^d$. An SCM with an underlying acyclic graph $\mathcal{G}$ specifies a data-generating process through structural equations
$$
X_i = f_i(\mathbf{X}_{pa_i}, \varepsilon_i), \qquad i = 1, \dots, d,
$$
where $f_i$ are causal (link) functions, $pa_i$ denotes the set of parents (direct causes) of $X_i$ in $\mathcal{G}$, and $\varepsilon_i$ are jointly independent noise variables. The central objective of causal discovery is to recover the causal structure, represented by the graph $\mathcal{G}$. If $\mathcal{G}$ and the conditional distributions are known, the joint distribution of $\mathbf{X}$ follows directly. Here we face the inverse task: given only the distribution of $\mathbf{X}$ (or a sample from it), we aim to infer $\mathcal{G}$. This is generally impossible without imposing additional assumptions on the underlying SCM \citep{Elements_of_Causal_Inference}.

The existing literature in the field presents numerous methods and corresponding results for causal discovery under various assumptions on the SCM \citep{ZhangReview}. When observing multiple environments following different interventions, the assumptions can be significantly less restrictive \citep{Peters_invariance, Multiple_contexts_Mooij, Causal_discovery_IV}. However, if the goal is to uncover causal relationships based solely on an observed random sample, the assumptions become more strict; typically assuming additive noise \citep{Lingam, Peters2014, reviewANMMooij, montagna23a}. This assumption of additivity $X_i = f_i(\textbf{X}_{pa_i}) + \varepsilon_i$ suggests that $\textbf{X}_{pa_i}$ influences only the mean of $X_i$, while the tail, variance, and higher moments remain fixed. This is a strong assumption, as the tail or other characteristics of the random variable can provide different information about the causal structure.

In this paper, we develop a framework where $\textbf{X}_{pa_i}$ can arbitrarily affect the mean, variance, tail, or other characteristics of $X_i$. However, a caution has to be taken because if the model is too general, the causal structure will become unidentifiable, meaning that multiple causal structures could produce the same distribution of $\textbf{X}$.
\begin{example}\label{Gaussian case}
    A useful model that allows the parent variables to influence both the mean and variance of \( X_i \) is given by the structural equation \( X_i = \mu(\textbf{X}_{pa_i}) + \sigma(\textbf{X}_{pa_i})\, \varepsilon_i \), where \( \varepsilon_i\) is Gaussian. Equivalently, the model for the conditional distribution is    \[
    X_i \mid \textbf{X}_{pa_i} \sim \mathcal{N}\big(\mu(\textbf{X}_{pa_i}),\ \sigma^2(\textbf{X}_{pa_i})\big).
    \]
\end{example}
\begin{example}\label{example_Poisson}
In certain applications, it may be reasonable to assume
$$X_i\mid \textbf{X}_{pa_i}\sim Poisson\big(\theta(\textbf{X}_{pa_i})\big),$$where $\theta$ is a function describing the rate of certain phenomena. Such a model is common in applications when $X_i$ represents a number of events occurring in a certain time period. 
\end{example}

We introduce a causal model (we call it the conditionally parametric causal model or CPCM) where the structural equation has the following form: 
\begin{equation}\label{11}\begin{split}
&X_i=f_i(\textbf{X}_{pa_i}, \varepsilon_i) = F^{-1}\big(\varepsilon_i; \theta(\textbf{X}_{pa_i})\big),\,\,\,\,\,\,\varepsilon_i\sim U(0,1), \\&   
\,\,\,\,\,\,\,\,\,\,\,\,\text{ or equivalently } X_i\mid \textbf{X}_{pa_i}\sim F\big(\theta(\textbf{X}_{pa_i})\big), \end{split}
\end{equation}
where $F$ is a known distribution function with a vector of parameters $\theta(\textbf{X}_{pa_i})$. 

\subsection{Setup and notation}
\label{Setup}

We adapt the usual notation of graphical models  (e.g., \citealp{PCalgorithm}). We consider a DAG (directed acyclic graph) $\mathcal{G}=(V,E)$ with a finite set of vertices (nodes) $V=\{1, \dots, d\}$ and a set of directed edges $E$, and write $pa_i(\mathcal{G})$, $ch_i(\mathcal{G})$ and $an_i(\mathcal{G})$ for parents, children and ancestors of the node $i$, respectively. In addition, we say that the node $i\in V$ is a source node if $pa_i(\mathcal{G})=\emptyset$, notation $i\in Source(\mathcal{G})$. Given a random vector $\textbf{X }= (X_i)_{i\in V}$ over some probability space with distribution $F_\textbf{X}$, we identify the vertices $j \in V$ with the variables $X_j$.  We omit the argument $\mathcal{G}$ if evident from the context.

We frequently use the concept of an exponential family, which is a class of probability distributions whose probability density function can be expressed as: 
\begin{equation}\label{Exponential family of distributions}
p(x;\theta) = h_1(x)h_2(\theta)e^{\sum_{i=1}^q\theta_iT_i(x)},
\end{equation}
where $h_1, h_2, T_i$ are measurable functions. We call $T_i$ a \textit{sufficient} statistic, $h_1$ a base measure, and $h_2$ a normalizing function. Note that $T_i$ are only unique up to a linear transformation.   Many well-known distribution families belong to the exponential family, including the Gaussian, Poisson, Binomial, and Gamma distributions. We assume that $q$ is minimal in the sense that we cannot write $p(x;\theta)$ using only $q-1$ parameters; see Appendix~\ref{appendix_exponential_family} that provides more information and detailed description. 

We use capital $F$ for distributions and small $p$ for densities. A random variable $Z$ that is uniformly distributed on $(0,1)$ is denoted as $Z\sim U(0,1)$. Support of a random variable $Z$ is denoted as $supp(Z)$.  We denote a random vector $\textbf{X}_S = \{X_s{:}\,\, s\in S\}$ for $S\subseteq V$.

\subsection{Related work}

Many papers address the problem of the identifiability of the causal structure (for a review, see \cite{ZhangReview}). 
\cite{Lingam} show identifiability for the linear non-Gaussian additive models (LiNGaM), where $X_i=\beta\textbf{X}_{pa_i} +\varepsilon_i$ for non-Gaussian noise variables $\varepsilon_i$. 
\cite{BuhlmannCAM} explore causal additive models (CAM) of the form $X_i = \sum_{j\in pa_i} g_j(X_j) +  \varepsilon_i$ for smooth functions $g_j$.  
\cite{hoyer2009} and \cite{Peters2014} develop a framework for additive noise models (ANM), where $X_i = g(\textbf{X}_{pa_i}) +  \varepsilon_i$. Under certain (not too restrictive) conditions on $g$, the authors show the identifiability of such models  \citep[Corollary 31]{Peters2014} and propose an algorithm estimating $\mathcal{G}$ (for a review on ANM, see \cite{reviewANMMooij}). 
All these frameworks assume that the variance of $X_i\mid \textbf{X}_{pa_i}$ does not depend on $\textbf{X}_{pa_i}$. This is a crucial aspect of the identifiability results.

\cite{Zhang2009} introduce a generalization known as the post-nonlinear model, defined by $
X_i = g_1\big(g_2(\textbf{X}_{pa_i}) +  \varepsilon_i\big),
$
with an invertible link function $g_1$.  
 \cite{ParkPoisson, ParkVariance} reveal identifiability in discrete models in which  $var[X_i\mid \textbf{X}_{pa_i}]$  is a quadratic function of $\mathbb{E}[X_i\mid \textbf{X}_{pa_i}]$. If $X_i\mid \textbf{X}_{pa_i}$ has a Poisson or binomial distribution, such a condition is satisfied.  They also provide an algorithm based on comparing dispersions for estimating a DAG in polynomial time. Other algorithms have also been proposed, with comparable speed and different assumptions on the conditional densities \citep{PolynomialTimeAlgorithmCausalGraphs}. \cite{Galanti, 10.24963/ijcai.2024/907} consider the neural SCM with representation $X_i = g_1\big(g_2(\textbf{X}_{pa_i}), \varepsilon_i\big),$ where $g_1$ and $g_2$ are assumed to be neural networks. 

Recently, location-scale models of the form 
\( X_i = g_1(\textbf{X}_{pa_i}) +  g_2(\textbf{X}_{pa_i})\varepsilon_i \)
have garnered attention. \cite{immer2022identifiability} demonstrated that bivariate non-identifiable location-scale models must satisfy a specific differential equation. \cite{strobl2022identifying} explored the problem of estimating patient-specific root causes in location-scale models. Additionally, \cite{Khemakhem_autoregressive_flows} provided more detailed identifiability results under Gaussian noise $\varepsilon_i$ in the bivariate case using autoregressive flows. \cite{xu2022inferring} investigated a more restricted location-scale model, dividing the range of the predictor variable into a finite set of bins and fitting an additive model in each bin. \cite{klippert2025skewnessrobustcausaldiscoverylocationscale} considered causal discovery in skewed location-scale models. 

Further, several different algorithms for estimating causal graphs have been proposed, working with different assumptions \citep{IGCI, Score-based_causal_learning, Slope,  Natasa_Tagasovska, krali2025causaldiscoveryheavytailedlinear}.  They are often based on Kolmogorov complexity or independence between certain functions in a deterministic scenario. 

Constraint-based methods, like the PC and FCI algorithms \citep{PCalgorithm, FCI}, are considered a gold standard for causal discovery. They utilize sequential independence testing for causal discovery, consistently estimating the Markov equivalence class. While these methods are powerful, they rely heavily on the accuracy of the conditional independence tests, making them sensitive to statistical errors and often resulting in many edges remaining unoriented.

A few authors assume that causal Markov kernels lie in a parametric family of distributions. \cite{JanzingSecondOrderExponentialModels} consider the case in which the density of  $X_i\mid \textbf{X}_{pa_i}$  lies in a second-order exponential family and the variables are a mixture of discrete and continuous random variables. \cite{ParkGHD} concentrate on a specific subclass of model (\ref{11}), where $F$ lies in a discrete family of generalized hypergeometric distributions; that is, the family of random variables in which the mean and variance have a polynomial relationship. To the best of our knowledge, there does not exist any study in the literature, that provides identifiability results in the case in which $F$ lies in a general class of the exponential family. This is the focus of this paper. 

\textbf{The structure of the paper is as follows.} Section \ref{Section2} introduces the main definitions and motivation in a bivariate case. Section \ref{Section_identifiability} presents identifiability results for the causal structure in the bivariate case, and Section \ref{Section4} discusses the multivariate extension. In Section \ref{Section5}, we propose an algorithm for estimating the causal graph under assumption (\ref{11}). Section \ref{simulations_section} contains an extensive simulation study. We provide three appendices: Appendix \ref{Appendix_A} includes formal definition of Exponential family and some omitted technical content; Appendix \ref{Appendix_simulations} details the experiments and Appendix \ref{SectionProofs} contains all proofs.

\section{Bivariate Conditionally Parametric Causal Models}
\label{Section2}

We focus on the bivariate SCM in this section, with multivariate extensions in Section~\ref{Section4}. The following definition describes the restriction on the SCM, assuming \(X_2 \mid X_1\) has the conditional distribution \(F\) with parameters \(\theta(X_1) \in \mathbb{R}^q\) for some \(q \in \mathbb{N}\).

\begin{definition}
We define the bivariate \textbf{conditionally parametric causal model} (bivariate $CPCM(F)$) with graph \(X_1 \to X_2\) by two assignments:
\begin{equation}\label{BCPCM}
X_1 = \varepsilon_1, \,\,\,\,\,\,\,\,\,\,\,\,\,\,\,\,\,X_2 =  F^{-1}\big(\varepsilon_2; \theta(X_1)\big),
\end{equation}
where \(\varepsilon_1 \indep \varepsilon_2\) are noise variables, \(\varepsilon_2\) is uniformly distributed, and \(F^{-1}\) is the quantile function of a distribution function $F$ with \(q\) parameters \(\theta(X_1) = \big(\theta_1(X_1), \dots, \theta_q(X_1)\big)^\top\).

We assume that \(\theta_i\) represent measurable functions, where at least one of the functions $\theta_1, \dots, \theta_q$ is non-constant on the support of \(X_1\). 
\end{definition}

We impose no restrictions on the marginal distribution of the cause. Note that we implicitly assume causal minimality \citep{zhang2010intervention}, as we assume that \(\theta\) is non-constant.

The Gaussian model introduced in Example~\ref{Gaussian case} is equivalent to the $CPCM(F)$ model with \(F\) being the Gaussian distribution and \(\theta(X_1) = (\mu(X_1), \sigma(X_1))^\top\). 

\subsection{\(CPCM(F_1, F_2, \dots, F_k)\) Models}

\subsubsection{Motivation}
\label{motivation_section}
Occam's razor posits that \(F_{\text{effect} \mid \text{cause}}\) should be ``simpler'' than \(F_{\text{cause} \mid \text{effect}}\). In model-based approaches for causal discovery, we define a ``simple distribution'' as one that belongs to a pre-defined class of distributions \(\mathcal{F}\). In ANM \citep{Peters2014}, \(\mathcal{F}\) consists of all distributions that can be expressed as the sum of a function of the cause and a noise term. In CPCM, \(\mathcal{F}\) is a given parametric family of  distributions. If \(\mathcal{F}\) is sufficiently small, we achieve identifiability of the causal graph \(\mathcal{G}\) because both \(F_{\text{effect} \mid \text{cause}}\) and \(F_{\text{cause} \mid \text{effect}}\) cannot lie in \(\mathcal{F}\).

However, the choice of \(\mathcal{F}\) is crucial, especially when dealing with mixtures of discrete and continuous distributions. Suppose we observe data as shown in Figure~\ref{Asymmetrical_picture}. To handle such cases, \(\mathcal{F}\) needs to include both continuous and discrete distributions. If we define \(\mathcal{F}\) as a class of \(CPCM(F)\) with continuous \(F\) (e.g., Gaussian), then \(F_{X_2 \mid X_1}\) can never lie in \(\mathcal{F}\). Conversely, choosing discrete \(F\) leads to \(F_{X_1 \mid X_2} \not\in \mathcal{F}\).

To accommodate a wide range of applications with various conditional distributions, we define \(\mathcal{F}\) as the union of \(CPCM(F_1), \dots, CPCM(F_k)\) models. By selecting \(F_1, \dots, F_k\) as a collection of ``standard simple well-known distributions'' (such as Gaussian, Gamma, Poisson, etc., see Section~\ref{Section5Model_choice}), \(\mathcal{F}\) is composed of ``standard simple (conditional) distributions'' with a wide range of possible supports, forms, and properties. We refer to this as the \(CPCM(F_1, F_2, \dots, F_k)\) model.

\begin{figure}[ht]
\centering
\includegraphics[scale=0.7]{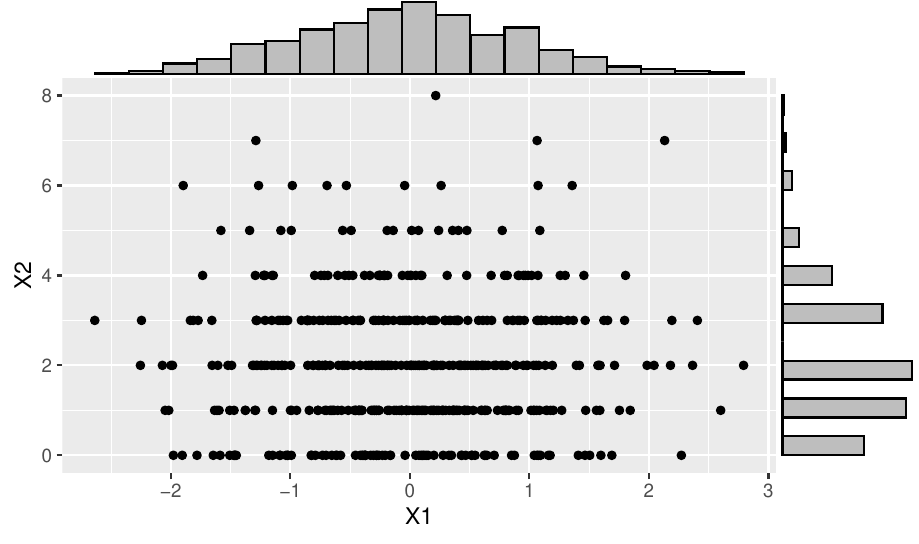}
\caption{A dataset generated as follows: \(X_1 \sim N(0,1)\), \(X_2 \sim \text{Poisson}(|X_1|)\); in other words, $X_1, X_2$ follow \(CPCM(F)\) with a Poisson \(F\) and DAG  \(X_1 \to X_2\).}
\label{Asymmetrical_picture}
\end{figure}

\subsubsection{Definition}\label{Section2.2.2}

\begin{definition}\label{CPCM(F1F2)}
Let \(F_1, \dots, F_k\) be a collection of distribution functions, each parameterized by a \(q_i\)-dimensional parameter, where \(q_i \in \mathbb{N}\) and \(i = 1, \dots, k\). A pair of dependent random variables \((X_1, X_2)\) follows the \(CPCM(F_1, F_2, \dots, F_k)\) model if there exists an \(i \in \{1, \dots, k\}\) such that \((X_1, X_2)\) follows a \(CPCM(F_i)\) model. Specifically, either:
\begin{equation*}
\begin{split}
 &  X_1 = \varepsilon_1, \quad X_2 = F_i^{-1}\big(\varepsilon_2; \theta_2(X_1)\big), \quad \varepsilon_2 \sim U(0,1), \quad \varepsilon_1 \indep \varepsilon_2, \quad \text{or} \\ 
 &  X_2 = \varepsilon_2, \quad X_1 = F_i^{-1}\big(\varepsilon_1; \theta_1(X_2)\big), \quad \varepsilon_1 \sim U(0,1), \quad \varepsilon_1 \indep \varepsilon_2,
\end{split}
\end{equation*}
for some \(i \in \{1, \dots, k\}\), where \(\theta_1(\cdot)\) and \(\theta_2(\cdot)\) are suitable parameter functions of dimension $q_i$, that are measurable and non-constant.
\end{definition}
We can potentially combine other well-known models, such as ANM and discrete QVF models; however, we will focus solely on \(CPCM\) classes for the remainder of this paper. Note the distinction between ANM and CPCM models: the former assumes that \(\theta(X)\) corresponds to the mean, while the latter allows \(\theta(X)\) to represent any distributional characteristic; however, CPCM imposes additional assumptions on the noise.

Extending the definition of \(CPCM(F)\) to allow multiple data-generating mechanisms induces the risk of unidentifiability. If the class \(CPCM(F_1, \dots, F_k)\) is too large, both \(F_{\text{effect} \mid \text{cause}}\) and \(F_{\text{cause} \mid \text{effect}}\) may lie within it. In the following section, we show that this is typically not the case as long as \(F_1, \dots, F_k\) belong to the exponential family of distributions.

\section{Identifiability results}
\label{Section_identifiability}
Identifiability is a prerequisite for causal discovery. We now examine the identifiability of the causal graph: can the true causal structure be inferred from the joint distribution under our CPCM model? 
\begin{definition}[Identifiability]\label{DEFidentifiability}
Let $F_{(X_1, X_2)}$ be a distribution that has been generated according to the $CPCM(F_1,\dots, F_k)$ model with graph $X_1\to X_2$. We say that the causal graph is identifiable from the joint distribution (equivalently, that the model is identifiable) if there does \textit{not} exist 
$\tilde{\theta}$ and a pair of random variables $\tilde{\varepsilon}_2\indep\tilde{\varepsilon}_1$, where $\tilde{\varepsilon}_1$ is uniformly distributed, such that the model $X_2=\tilde{\varepsilon}_2,  X_1=  F_i^{-1}\big(\tilde{\varepsilon}_1;\tilde{\theta}(X_2)\big)$ for some $i\in\{1, \dots, k\}$ generates the same distribution $F_{(X_1,X_2)}$. 
\end{definition}

\subsection{Identifiability in $CPCM(F)$}

First, we discuss the Gaussian case. Recall that in the additive Gaussian model, where $X_2=f(X_1)+\varepsilon_2$, $\varepsilon_2\sim N(0, \sigma^2)$, the identifiability holds if and only if $f$ is non-linear \citep{hoyer2009}. We provide a different result with both mean \textit{and} variance as functions of the cause. A similar result is found in \cite[Theorem 1]{Khemakhem_autoregressive_flows} in the context of autoregressive flows and where only a sufficient condition for identifiability is provided. Another similar problem is studied in \cite{immer2022identifiability} and \cite{strobl2022identifying}, both of which discuss identifiability in general location-scale models. 

\begin{theorem}[Gaussian case]\label{normalidentifiability}
Let $(X_1,X_2)$ admit the $CPCM(F)$ model with graph $X_1\to X_2$, where $F$ is the Gaussian distribution function with parameters $\theta(X_1)=\big(\mu(X_1), \sigma(X_1)\big)^\top$ as in Example \ref{Gaussian case}.

Let $p_{\varepsilon_1}$ be the density of $\varepsilon_1$, which is absolutely continuous with full support $\mathbb{R}$. Let $\mu(x), \sigma(x)$ be two times differentiable.  Then, the causal graph is identifiable from the joint distribution if and only if there do not exist $a,c,d,e, \alpha, \beta\in\mathbb{R}$,  
$a\geq 0,c>0, \beta>0$, such that
\begin{equation}\label{norm}
\frac{1}{\sigma^2(x)}=ax^2 + c, \,\,\,\,\,\,\,\,\,\,\,\,\,\,\,\,\,\,\,\,\, \frac{\mu(x)}{\sigma^2(x)}=d+ex,
\end{equation}
for all $x\in\mathbb{R}$ and
\begin{equation}\label{DensityDEF}
p_{\varepsilon_1}(x) \propto \sigma(x)e^{-\frac{1}{2}\big[ \frac{(x-\alpha)^2}{\beta^2}  - \frac{\mu^2(x)}{\sigma^2(x)}\big]},
\end{equation}
where $\propto$  represents an equality up to a constant (here, $p_{\varepsilon_1} $  is a valid density function if and only if $\frac{1}{\beta^2}>  \frac{e^2}{c}\mathbbm{1}[a=0]$). 
Specifically, if $\sigma(x)$ is constant (case $a=0$), then the causal graph is identifiable, unless $\mu(x)$ is linear and $p_{\varepsilon_1}$ is the Gaussian density. 
\end{theorem}

The proof is provided in \hyperref[Proof of normalidentifiability]{Appendix} \ref{Proof of normalidentifiability}. Moreover, a visual example of an unidentifiable Gaussian model with $a=c=d=e=\alpha = \beta=1$ can be found in  \hyperref[Proof of normalidentifiability]{Appendix} \ref{Proof of normalidentifiability}, Figure \ref{GaussianDensity}. 

Theorem \ref{normalidentifiability} indicates that the non-identifiability holds only in the ``special case,'' when $\frac{\mu(x)}{\sigma^2(x)}, \frac{-1}{2\sigma^2(x)}$ are linear and quadratic, respectively. Note that natural parameters of a Gaussian distribution are  $\frac{\mu}{\sigma^2}, \frac{-1}{2\sigma^2}$, and sufficient statistics of the Gaussian distribution have a linear and quadratic form (for the definition of the exponential family, natural parameter and sufficient statistic, see \hyperref[appendix_exponential_family]{Appendix} \ref{appendix_exponential_family}). We show that such connections between non-identifiability and sufficient statistics hold in the more general context of the exponential family.  

\begin{proposition}[General case, one parameter]\label{Necessary condition for identifiability}
Let $q=1$. Let $(X_1, X_2)$ admit the $CPCM(F)$ model with graph $X_1\to X_2$, where $F$ lies in the exponential family of distributions with a sufficient statistic $T$. The causal graph is identifiable if at least one of the following conditions is met:
\begin{enumerate}
\item Supports of $X_1$ and $X_2$ differ. 
    \item There do not exist $a,b\in\mathbb{R}$, such that
 \begin{equation}\label{eq000}
     \theta(x)=a\, T(x)+b,\,\,\,\,\,\,\,\,\,\,\,\,\, \forall x\in supp(X_1). 
 \end{equation} 
 \item There does not exist $c\in\mathbb{R}$, such that 
\begin{equation}\label{eq007}
p_{X_1}(x)\propto \frac{h_1(x)}{h_2[\theta(x)]}e^{cT(x)}, \,\,\,\,\,\,\,\,\,\,\,\,\,\,\,\,\,\,\,\, \forall x\in supp(X_1),
\end{equation}
where $h_1$ is a base measure of $F$ and $h_2$ is the normalizing function of $F$ defined in \hyperref[appendix_exponential_family]{Appendix} \ref{appendix_exponential_family}. 
\end{enumerate}
\end{proposition}
\begin{hproof} While this proposition follows as a special case of Theorem~\ref{thmAssymetricMultivariatesufficient}, we outline the key ideas behind the proof.

If the graph is \textit{not} identifiable, there exists a function $ \tilde{\theta}$, such that 
causal models $X_1 = \varepsilon_1, X_2 = F^{-1}\big(\varepsilon_2; \theta(X_1)\big)$, and $X_2 = \varepsilon_2, X_1 = F^{-1}\big(\varepsilon_1; \tilde{\theta}(X_2)\big)$ generate the same joint distribution.

\textbf{1)} If the supports of $X_1$ and $X_2$ differ, then $X_1$ trivially can not be written as $X_1 = F^{-1}\big(\varepsilon_1; \tilde{\theta}(X_2)\big)$  since the support of any distribution in the exponential family is fixed. This immediately rules out non-identifiability in such cases. At the population level, even a difference on a measure-zero set is sufficient to ensure identifiability.

\textbf{2) } Assuming the causal graph is not identifiable, we decompose the joint density:
\begin{equation}\label{eq69}
  p_{(X_1, X_2)}(x,y) = p_{X_1}(x)p_{X_2\mid {X_1}}(y\mid x) = p_{X_2}(y)p_{{X_1}\mid {X_2}}(x\mid y),\,\,\,\,\,\,\,\, x,y\in supp(X_1).
 \end{equation}
Since $F$ belongs to the exponential family, we rewrite the conditional densities using notation from \eqref{Exponential family of distributions}:$$p_{X_2\mid {X_1}}(y\mid x) = h_{1}(y)h_{2}[\theta(x)]\exp[\theta(x)T(y)], \,\,\,\,\,\,p_{{X_1}\mid {X_2}}(x\mid y) = h_{1}(x)h_{2}[\tilde{\theta}(y)]\exp[\tilde{\theta}(y)T(x)].$$
Substituting this into \eqref{eq69} gives:
\begin{equation}\label{eq098}
    \begin{split}
  &   p_{X_1}(x)h_{1}(y)h_{2}[\theta(x)]\exp[\theta(x)T(y)] =p_{X_2}(y) h_{1}(x)h_{2}[\tilde{\theta}(y)]\exp[\tilde{\theta}(y)T(x)],\\&
\underbrace{\log\bigg\{  \frac{ p_{X_1}(x)h_{2}[\theta(x)]}{h_{1}(x)}\bigg\}}_{f(x)}+ \underbrace{\log\bigg\{\frac{h_{1}(y)}{ p_{X_2}(y)h_{2}[\tilde{\theta}(y)]}\bigg\}}_{g(y)} = \tilde{\theta}(y)T(x)-\theta(x)T(y),
\end{split}
\end{equation}
where the second equation follows from taking the logarithm of both sides after dividing by  $h_1$ and $h_2$. Differentiating both sides with respect to $x$ and $y$, and fixing $y$ such that $T'(y)\neq 0$, we obtain $\theta'(x) = \frac{\tilde{\theta}'(y)}{T'(y)}T'(x)$. Integrating this equation with respect to $x$ leads to (\ref{eq000}). Note that we do not need to assume differentiability of $\theta$, and we can get around it by applying Lemma \ref{PomocnaLemma1}.

\textbf{3) } equality (\ref{eq098}) implies 
$f(x) + g(y) = a_1T(x) + a_2T(y) +a_3$ for some constants $a_1, a_2, a_3\in\mathbb{R}$. Therefore, fixing $y$ yields $f(x)=\log\bigg\{  \frac{ p_{X_1}(x)h_{2}[\theta(x)]}{h_{1}(x)}\bigg\} = a_1T(x) + const$. Rewriting this directly yields (\ref{eq007}). 
\end{hproof}

Intuitively, condition~\eqref{eq000} rules out joint distributions that are symmetric with respect to the transformed axis $aT(x)+b$, similar to how symmetry around $y = a+bx$ causes non-identifiability in the Gaussian ANM case \citep{Zhang2009}.
In the $CPCM(F)$ setting, the parameter $\theta(x)$ affects the distribution through the sufficient statistic $T(x)$.
If $\theta(x)$ takes the form $T(x)$ (or $aT(x)+b$, since a sufficient statistic is defined only up to an affine transformation), the joint law has this symmetry, making both causal directions consistent with the data.
By violating \eqref{eq000}, the symmetry is broken and the true causal direction becomes identifiable.

As a consequence of Proposition \ref{Necessary condition for identifiability}, we extend the results demonstrated by \cite{ParkPoisson} and \cite{ParkGHD} for a Poisson DAG model. These authors established the identifiability of a Poisson DAG model, where all variables (including source variables) given their parents follow a Poisson distribution. We present an analogous result, relaxing the restriction on the source variables.

\begin{consequence}\label{paretoidentifiability}
\begin{itemize}
    \item Let $(X_1,X_2)$ admit the $CPCM(F)$ model with graph $X_1\to X_2$, where $F$ is the Poisson distribution function with rate $\lambda$.  Then, the causal graph is \textit{not} identifiable if and only if
\begin{equation*}
 \lambda(x) =e^{ax+b},\,\,\,\,\,\,\,\, P(X_1=x) \propto \frac{e^{ \lambda(x)+cx}}{x! }, \,\,\,\,\,\,\,\,\forall x\in \{0,1,2, \dots \},
 \end{equation*}
for some $a<0,b,c\in\mathbb{R}$.
    \item Let $(X_1,X_2)$ admit the $CPCM(F)$ model with graph $X_1\to X_2$, where $F$ is the Pareto distribution function. Then, the causal graph is \textit{not} identifiable if and only if 
\begin{equation*}
\theta(x) = a\log(x) +b,\,\,\,\,\,\,\,\, p_{X_1}(x) \propto \frac{1}{ [a\log(x)+b] x^{c+1} }, \,\,\,\,\,\,\,\,\forall x\geq 1,
\end{equation*}
for some $a,b,c>0$. 

    \item Let $(X_1,X_2)$ admit the $CPCM(F)$ model with graph $X_1\to X_2$, where $F$ is Bernoulli distribution function. Then, the causal graph is identifiable if and only if $supp(X_1) \neq supp(X_2) $.
\end{itemize}
\end{consequence}
The proof is provided in \hyperref[Proof of pareto identifiability]{Appendix} \ref{Proof of pareto identifiability}, together with definitions of the distribution functions. Observe that if $a=0$, then $X_1\indep X_2$ and the (empty) graph is trivially identifiable. 

Note that in the first two bullet points of Consequence~\ref{paretoidentifiability}, we have three free parameters: \(a\), \(b\), and \(c\). The non-identifiability of the graph in a Bernoulli model arises from the fact that the joint distribution of \((X_1, X_2)\) can be fully characterized by only three parameters.

\subsection{Identifiability in $CPCM(F_1, \dots, F_k)$ models}

We generalize Proposition \ref{Necessary condition for identifiability}  to the general case. The following theorem establishes that $CPCM(F_1, \dots, F_k)$ models are ``typically'' identifiable, except for a finite-dimensional set within the space of all possible distributions.

\begin{theorem}
\label{thmAssymetricMultivariatesufficient}
Let $(X_1, X_2)$ follow the $CPCM(F_1, \dots, F_k)$ model with graph $X_1\to X_2$, where $F_1, \dots, F_k$ belong to the exponential family of distributions with corresponding sufficient statistics  
$T_m= (T_{m,1}, \dots, T_{m,q_m})^\top$, $m=1, \dots, k$.  Following Definition~\ref{CPCM(F1F2)}, 
let $\tilde{m}\in\{1, \dots, k\}$ be the index such that  
$X_2 = F_{\tilde{m}}^{-1}\big(\varepsilon_2; 
\theta_2(X_1)\big)$. 

The causal graph is identifiable if for all $m\in \{1, \dots, k\}$, at least one of the following holds: 
\begin{itemize}
    \item $ supp(F_m) \neq supp(X_1)$. 
    \item The function \( \theta_2 \) is not a linear combination of the sufficient statistics \( T_{m,1}, \dots, T_{m,q_m} \), i.e., there do not exist coefficients \( a_{i,j}, b_i \in \mathbb{R} \) for \( i = 1, \dots, q_{\tilde{m}} \) and \( j = 1, \dots, q_m \) such that  
   \begin{equation}\label{eq158}
   \theta_{2,i}(x) = \sum_{j=1}^{q_m} a_{i,j} T_{m,j}(x) + b_i, \quad \forall x \in \operatorname{supp}(X_1), \quad \forall i \in \{1, \dots, q_{\tilde{m}}\}.
   \end{equation}  
    \item There do not exist constants \( c_1, \dots, c_{q_m} \in \mathbb{R} \) such that the density of \( X_1 \) satisfies  
   \begin{equation}\label{eq007v2}
  p_{X_1}(x) \propto \frac{h_{m,1}(x)}{h_{\tilde{m},2}[\theta_2(x)]} e^{\sum_{i=1}^{q_m} c_i T_{m,i}(x) }, \quad \forall x \in \operatorname{supp}(X_1),
   \end{equation}  where \( h_{m,1} \) is a base measure associated with \( F_{m} \) and \( h_{\tilde{m},2} \) is the normalizing function of \( F_{\tilde{m}} \), both defined in \hyperref[appendix_exponential_family]{Appendix} \ref{appendix_exponential_family}. 
\end{itemize}
Consequentially, the space of non-identifiable distributions is contained in a $\tilde{d}$-dimensional space, where 
\begin{equation}\label{dimension_in_theorem2}
    \tilde{d} \leq \sum_{m\in\{1, \dots, k\}:  supp(F_m) = supp(X_1)} (q_m+1)(q_{\tilde{m}}+1) -1 .  
\end{equation}

\end{theorem}

The proof is provided in \hyperref[Proof of thmAssymetricMultivariatesufficient]{Appendix} \ref{Proof of thmAssymetricMultivariatesufficient}. It is insightful to examine the dimension of the space of unidentifiable distributions for different choices of \( F \).  In one-parameter cases, such as in Consequence~\ref{paretoidentifiability}, the set of all unidentifiable distributions lies within a three-dimensional space, similar to the case for ANM \cite[Proposition 21]{Peters2014}. For the Gaussian \( CPCM(F) \) model, Theorem~\ref{normalidentifiability} shows that this dimension is \( 6 \), despite Theorem~\ref{thmAssymetricMultivariatesufficient} initially suggesting \( (2+1)(2+1) -1 = 8 \). In the proof of Theorem~\ref{normalidentifiability}, we showed that two of these coefficients must be zero, confirming that \eqref{dimension_in_theorem2} provides only an upper bound and the actual dimension can be smaller.  

Since the space of all distributions is infinite-dimensional (assuming infinite support), one can argue that identifiability holds for ``most distributions,'' regardless of the choice of \( F \). However, when \( F \) has many parameters, the model often lies ``close'' to an unidentifiable case, making the finite sample inference significantly more challenging.  

\begin{consequence}\label{consequenceprva}
Suppose that \( \text{supp}(X_1) = \mathbb{R} \), \(\text{supp}(X_2) = \{0, 1, \dots\}\) such as on Figure~\ref{Asymmetrical_picture}, and let \((X_1, X_2)\) admit the \(CPCM(F_1, F_2)\) model with graph \(X_1 \to X_2\), where \(F_1\) is a Gaussian distribution and \(F_2\) is a Poisson distribution with rate parameter \(\lambda\). The causal graph is identifiable if and only if there do not exist constants \(a_1, a_2, b, c_1, c_2\in \mathbb{R}\), $a_1, c_1<0$, such that:
\begin{equation*}
\lambda(x) = e^{a_1 x^2 +a_2x + b}, \quad p_{X_1}(x) \propto e^{c_1 x^2 + c_2 x },\quad\quad \forall x \in \mathbb{R}.
\end{equation*}
Details and more examples are provided in  \hyperref[consequence]{Appendix} \ref{consequence}.  
\end{consequence}

\section{Multivariate case $d\geq 2$}
\label{Section4}

We extend the theory to the case with possibly more than two variables, $\textbf{X} = (X_1, \dots, X_d)^\top$.  

\begin{definition}\label{DefinitionCPCM}
Let \(\{F_1, \dots, F_k\}\) be a collection of distribution functions with \(q_1, \dots, q_k\) parameters, respectively, where  \(q_1, \dots, q_k\in\mathbb{N}\). We define a conditionally parametric causal model \(CPCM(F_1, \dots, F_k)\) with an underlying DAG \(\mathcal{G}\) as a collection of equations:
\begin{equation*}
  X_j = \begin{cases}
    \varepsilon_j, & \text{if } j \in \text{Source}(\mathcal{G}), \\
    F_{\pi(j)}^{-1}\big(\varepsilon_j; \theta_j(\mathbf{X}_{pa_j})\big), \text{where } \pi(j) \in \{1, \dots,k\},\,\,\,\,\,\,\, & \text{if } j \notin \text{Source}(\mathcal{G}),
  \end{cases}
\end{equation*}
where \((\varepsilon_1, \dots, \varepsilon_d)^\top\) is a collection of jointly independent random variables with \(\varepsilon_j \sim \text{U}(0, 1)\) for all \(j \notin \text{Source}(\mathcal{G})\), and \(\theta_j\) are non-constant functions in any of their arguments.
\end{definition}
Simply said, we assume that $X_j\mid \textbf{X}_{pa_j}$ is distributed according to distribution $F_{\pi(j)}$ with parameters $\theta_j(\textbf{X}_{pa_j})$, where $F_{\pi(j)}$ is either $F_1, F_2,  \dots$, or  $F_k$. 
Although we implicitly assume causal minimality \citep{zhang2010intervention}, we do not require the stronger assumption of faithfulness \citep{uhler2013geometry}. 

The question of the identifiability of $\mathcal{G}$ in the multivariate case is in order. Here, it is not satisfactory to consider the identifiability of each pair of $X_i\to X_j$ separately. Each pair $X_i, X_j$  needs to have an identifiable causal relation \textit{conditioned} on other variables $\textbf{X}_S$. Such an observation was first made by \cite{Peters2014} in the context of additive noise models. We now provide a more precise statement in the context of  $CPCM(F_1, \dots, F_k)$. 

\begin{definition} 
We say that the $CPCM(F_1, \dots, F_k)$ is \textit{pairwise identifiable}, if for all $ i,j\in V$, $ S\subseteq V$, such that $i\in pa_j$ and  $pa_j\setminus \{i\}\subseteq S \subseteq nd_j\setminus\{i,j\}$, there exists  $\textbf{x}_{S}{:}\,\,  p_S(\textbf{x}_S)>0$, which satisfies that a bivariate model defined as $X=\tilde{\varepsilon}_X, Y = F^{-1}_j\big(\tilde{\varepsilon}_Y, \tilde{\theta}(X)\big)$ is identifiable (in the sense of Definition \ref{DEFidentifiability}), where  $F_{\tilde{\varepsilon}_X} = F_{X_i\mid \textbf{X}_{S} =\textbf{ x}_S}    $ and $\tilde{\theta}(x) = \theta_j(\textbf{x}_{pa_j\setminus\{i\}}, x)$,  $x\in supp(X)$.
\end{definition}

 \begin{lemma}\label{thmMultivairateIdentifiability}
Let $F_{\textbf{X}}$ be generated by the pairwise identifiable $CPCM(F_1, \dots, F_k)$ with DAG $\mathcal{G}$. Then,  $\mathcal{G}$ is identifiable from the joint distribution. 
 \end{lemma}
The proof follows as a consequence of Theorem 28 in \cite{Peters2014} and is provided in \hyperref[Proof of thmMultivairateIdentifiability]{Appendix} \ref{Proof of thmMultivairateIdentifiability}.

\begin{consequence}[Multivariate Gaussian case]\label{ExampleMultivariateGaussiancase}
Suppose that $\textbf{X}=(X_1, \dots, X_d)$ follow $CPCM(F)$ with a Gaussian distribution function $F$. This corresponds to $X_j\mid \textbf{X}_{pa_j}\sim N\big(\mu_j(\textbf{X}_{pa_j}), \sigma_j^2(\textbf{X}_{pa_j})\big)$ for all $j=1, \dots, d$ and for some functions $\mu_j, \sigma_j$. In other words, we assume that the data-generation process has the following form: 
$$
X_j = \mu_j(\textbf{X}_{pa_j}) + \sigma_j(\textbf{X}_{pa_j})\,\varepsilon_j, \,\,\,\,\,\,\,\,\,\,\,\,\,\,\,\,\, \text{where  } \varepsilon_j \text{  is Gaussian.}
$$
Potentially, source nodes can have arbitrary distributions. Combining Theorem~\ref{normalidentifiability} and Lemma~\ref{thmMultivairateIdentifiability}, the causal graph $\mathcal{G}$ is identifiable if the 
functions $\theta_j(\textbf{x}):=\big(\mu_j(\textbf{x}), \sigma_j(\textbf{x})\big)^\top, \textbf{x}\in\mathbb{R}^{|pa_j(\mathcal{G})|}$ , $j=1, \dots, d$, are \textit{not} in the form (\ref{norm}) in any of their arguments. 
\end{consequence}

\section{Inference}
\label{Section5}

\subsection{ CPCM algorithm - bivariate case   }
\label{Section_Algorithm}

Our CPCM methodology is based on selecting an appropriate causal model (in our case, the choice of collection $\{F_1, \dots, F_k\}$) and a measure of a model fit. In the following subsections, we measure the model fit by exploiting the principle of independence between the cause and the mechanism.

We say that a DAG $\mathcal{G}$ is \textbf{plausible} under $CPCM(F)$ model if the joint distribution \textit{can} be generated via $CPCM(F)$ model with graph $\mathcal{G}$. The Algorithm~\ref{Algorithm1} describes the main steps to test the plausibility and estimation of $\mathcal{G}$ in the bivariate case. 

\begin{algorithm}[]
  \SetAlgoLined
  \KwData{ Random sample $(x_{1,1}, x_{2,1})^\top, \dots, (x_{1,n}, x_{2,n})^\top$}
  \KwResult{Estimate $\hat{\mathcal{G}}$ and plausibility of graphs $X_1\to X_2$ and $X_2\to X_1$}
\textbf{Step 0) }Test independence between $X_1$ and $X_2$. 

\textbf{Step 1) } Determine plausibility of   $X_1\to X_2$ using the following:

  \hspace{0.3cm}\textbf{1a)} Estimate $\theta(X_1)$ in $X_2 = F^{-1}(\varepsilon_2; \theta(X_1))$; compute $\hat{\varepsilon}_2 := F(X_2; \hat{\theta}(X_1))$.
  
  \hspace{0.3cm}\textbf{1b)} Test independence between $\hat{\varepsilon}_2$ and $X_1$. If the p-value is larger than $\alpha=0.05$, mark $X_1 \to X_2$ as \textit{plausible}.

  \textbf{Step 2:} Repeat Step 1 for $X_2 \to X_1$.

\textbf{Forced estimate}  (Choose the direction with the higher residual independence):  Return $X_1 \to X_2$ if $\text{p-value}(\hat{\varepsilon}_2, X_1) > \text{p-value}(\hat{\varepsilon}_1, X_2)$, else return $X_2 \to X_1$.

\textbf{Conservative Estimate:}  
\begin{itemize}
 \setlength{\itemsep}{0pt} % Reduce space between items
    \setlength{\parskip}{0pt} % Reduce space between paragraphs
    \item If Step 0 fails to reject independence, return $\hat{\mathcal{G}} = \emptyset$.
    \item If exactly one of the two graph directions is plausible, return it as $\hat{\mathcal{G}}$.
    \item If both directions are plausible, return \texttt{"Unidentifiable case"}.
    \item If neither is plausible, return \texttt{"Assumptions not fulfilled"}.
\end{itemize}

  \caption{CPCM(F) - bivariate case}
  \label{Algorithm1}
\end{algorithm}

An estimation of ${\theta}(X_1)$ in Step 1a) is discussed in detail in Appendix~\ref{Appendix_consistency}. It can be performed using any suitable machine learning method, such as GAM, GAMLSS, random forests, or neural networks \citep{GAM, GAMLSS}. For the independence test in Step 1b), one may use the HSIC test (kernel-based Hilbert–Schmidt independence criterion; \cite{Kernel_based_tests}) or a copula-based test \citep{copula_based_independence_test}.

The conservative estimate $\hat{\mathcal{G}}$ helps guard against unfulfilled assumptions or unidentifiable cases. If it returns ``Assumptions not fulfilled,'' it means that we were unable to fit the $CPCM(F)$ model in either direction, suggesting that the variables do not follow the $CPCM(F)$ model.  In this case, one should consider increasing the complexity (i.e., the number of parameters) of $F$. This is discussed further in Section~\ref{Section5Model_choice}. 

If it returns ``Unidentifiable case,'' this means that we were able to fit the $CPCM(F)$ model in both directions. This could indicate that the sample size is too small or that we are in an unfortunate unidentifiable case, such as the one described in Consequence~\ref{consequenceprva}.  In this case, one should consider decreasing the complexity (i.e., the number of parameters) of $F$.

While the warnings from the conservative estimate  $\hat{\mathcal{G}}$  are useful, we often require a single estimate of $\hat{\mathcal{G}}$  for comparison with other benchmark methods. 

\subsubsection{Extension to $CPCM(F_1, \dots, F_k)$}\label{section5.1.1}
The following adjustment to Step 1 in Algorithm~\ref{Algorithm1} can be applied to accommodate the \(CPCM(F_1, \dots, F_k)\) model, given a collection \(\{F_1, \dots, F_k\}\).

\begin{algorithm*}[H]
  \SetAlgoLined
$\boldsymbol{CPCM(F_1, \dots, F_k)}$

    \textbf{Step 1) } Determine plausibility of   $X_1\to X_2$ using the following: 

\textbf{$\,\,\,\,\,\,$Step 1a) } Estimate the set $S:=\{j\in \{1, \dots, k\}: supp(F_j) = supp(X_2)\}$. If empty, return ``STOP: Inappropriate choice of $F$''. 

\textbf{$\,\,\,\,\,\,$Step 1b) } For all  $j\in \hat{S}$,  estimate $\theta(X_1)$ in a model $X_2 = F_j\big(\varepsilon_2; \theta(X_1)\big)$,  and compute probability transform $\hat{\varepsilon}_2^j := F_j\big(X_2; \hat{\theta}(X_1)\big)$.  

\textbf{$\,\,\,\,\,\,$Step 1c)} Compute the p-value of an independence test between $\hat{\varepsilon}_2^j$ and $X_1$ for all $j\in\hat{S}$. Choose the largest p-value. Direction $X_1\to X_2$ is marked as \textit{plausible} if this p-value is larger than  $\alpha=0.05$. 

%  \caption{Modification of Step 1 in Algorithm~\ref{Algorithm1} accounting for $CPCM(F_1, \dots, F_k)$}
%  \label{Algorithm2}
\end{algorithm*}

By estimating the set $S$ in Step 1a), we can  preliminarily filter out the $F_j$ choices that are evidently unsuitable, such as fitting Gaussian model when the data are discrete. We can apply a simple heuristic under the assumption that \( \text{supp}(X_2) \) and \( \text{supp}(F_j) \) for all \( j \) are one of the following: 1) \(\mathbb{R}\) (Gaussian), 2) \(\mathbb{R}^+\) (Gamma), 3) \([0,1]\) (Beta), or 4) \(\mathbb{N}\)(Poisson). The heuristic is as follows: if the number of unique values in \( (x_{2,1}, \dots, x_{2,n}) \) is fewer than \( n/10 \) and the values lie in $\mathbb{N}$, we set \( \text{supp}(X_2) = \mathbb{N} \). Otherwise, if all values lie within \([0,1]\), we set \( \text{supp}(X_2) = [0,1] \). To distinguish between \(\mathbb{R}\) and \(\mathbb{R}^+\), we use the skewness of the distribution: if the skewness is close to 0, the distribution resembles a Gaussian distribution, so we set \( \text{supp}(X_2) = \mathbb{R} \). Otherwise, the distribution resembles a Gamma distribution, so we set \( \text{supp}(X_2) = \mathbb{R}^+ \).

\subsection{Multivariate case: $CPCM(F)$ as a score-based DAG optimization
\label{Section_score_based_algorithm}
}
The forced estimate of $\mathcal{G}$ in Algorithm~\ref{Algorithm1} can be seen as a score-based algorithm; defining the score of a graph as a p-value of the corresponding independence test, we simply compare the scores of graphs $X_1\to X_2$ and $X_2\to X_1$. In the following, we formalize this idea, leading to generalizing the CPCM inference to multivariate case. 

Following the ideas of \cite{Bregmans_information, Score-based_causal_learning, Peters2014}, we use the following penalized independence score: 
\begin{equation}\label{score_definition1}
\hat{\mathcal{G}} =  \argmin_{\mathcal{G}\in DAG(d)}s(\mathcal{G}) = \argmin_{\mathcal{G}\in DAG(d)}\rho (\hat{\varepsilon}_1, \dots, \hat{\varepsilon}_d) + \lambda (\text{Number of edges in }\mathcal{G}),
\end{equation}
where $\rho$ represents some dependence measure, $\lambda$ is a hyperparameter, $DAG(d)$ is the set of all DAGs over $V=\{1, \dots, d\}$ and $\hat{\varepsilon}_1, \dots, \hat{\varepsilon}_d$ are noise estimators obtained by estimating ${\theta}_i(\textbf{X}_{pa_i(\mathcal{G})})$ and putting $\hat{\varepsilon}_i := F\big(X_i; \hat{\theta}_i(\textbf{X}_{pa_i(\mathcal{G})})\big)$ analogously to Algorithm~\ref{Algorithm1}. 

With regard to choice of $\rho$, we use minus the logarithm of the p-value of the independence test \citep{copula_based_independence_test} and $\lambda = 2$. These choices appear to work well in practice, but we do not provide any theoretical justification of their optimality.

Analogously to the bivariate case, we can define that a DAG ${\mathcal{G}}$ is \textbf{plausible} if the p-value of the corresponding independence test between $(\hat{\varepsilon}_1, \dots, \hat{\varepsilon}_d)$ is larger than level $\alpha\in (0,1)$. If every ${\mathcal{G}}\in DAG(d)$ is not plausible, it suggests that the variables do not follow the $CPCM(F)$ model. In this case, one should consider increasing the complexity (i.e., the number of parameters) of $F$; this is discussed in more detail in Section~\ref{S1_or_S2}. On the contrary to the bivariate case, many DAGs can be plausible in an identifiable $CPCM(F)$ model. In particular, if a true causal graph $\mathcal{G}$ is plausible, than any $\tilde{\mathcal{G}}\supseteq \mathcal{G}$ should be plausible. 

 It is important to note that in a bivariate case, score based estimate \eqref{score_definition1} is equal to the forced output of Algorithm~\ref{Algorithm1}, given $\lambda = -\infty$.

\subsubsection{$CPCM(F_1,\dots, F_k)$ extension
}
Similarly to the extension of Algorithm~\ref{Algorithm1}, the following adjustment can be applied to accommodate the \(CPCM(F_1, \dots, F_k)\) model, given a collection \(\{F_1, \dots, F_k\}\):

\begin{equation}
\label{multi_extention}
s(\mathcal{G}) = \min_{j_1\in \hat{S}_1, \dots, j_d\in \hat{S}_d}\rho (\hat{\varepsilon}_1^{j_i}, \dots, \hat{\varepsilon}_d^{j_d}) + \lambda (\text{Number of edges in }\mathcal{G}),
\end{equation}
where $\hat{S}_i$ are estimates of $S_i:= \{ j\in\{1, \dots, k\}: supp(X_i) = supp(F_j) \}$, and where  $\hat{\varepsilon}_i^{j_i} := F_{j_i}\big(X_i; \hat{\theta}_i(\textbf{X}_{pa_i(\mathcal{G})})\big)$. If $supp(F_i) \neq supp(F_j)$ for all $i\neq j$, we end up with a single evaluation of the score function for each possible $\mathcal{G}$. Finally, analogously to $CPCM(F)$, we say that DAG ${\mathcal{G}}$ is plausible if there exist $j_1\in \hat{S}_1, \dots, j_d\in \hat{S}_d$ such that the p-value of the independence test between $(\hat{\varepsilon}_1^{j_i}, \dots, \hat{\varepsilon}_d^{j_d})$ is larger than level $\alpha\in (0,1)$.

\subsubsection{Consistency}

\cite{reviewANMMooij} demonstrates consistency of a causal discovery algorithm in ANMs.  We establish an analogous result for $CPCM(F_1, \dots, F_k)$. 

\begin{proposition}
\label{consistency_proposition}
Let $(X_1, X_2)$ follow an identifiable  $CPCM(F_1, \dots, F_k)$ with DAG $\mathcal{G}$. Then, our score based algorithm presented in Section~\ref{Section_score_based_algorithm} is consistent, meaning that
$$\hat{\mathcal{G}} \overset{P}{\to}\mathcal{G}\,\,\,as\,\,n\to\infty,$$
given that we employ a ``suitable'' estimation procedure for the estimation $\hat{\varepsilon}_i$, we use HSIC score as our choice of $\rho$ and consistent estimators $\hat{S}_i$ (for definitions, proof and more details, see  Appendix~\ref{Appendix_consistency}). 
\end{proposition}

Additionally, in the \textit{unidentifiable} case, Algorithm~\ref{Algorithm1} detects non-identifiability with high probability, as stated in the following lemma. 

\begin{lemma}[Algorithm~\ref{Algorithm1} under unidentifiability]
\label{Lemma3}
Let $(X_1, X_2)$ follow an unidentifiable model $CPCM(F_1, \dots, F_k)$. 
Assume that Algorithm~\ref{Algorithm1} employs the HSIC independence test in step~1b) and that the regression function is estimated perfectly, i.e., using oracle estimator $\hat{\theta} = \theta$. 
Then, for sufficiently large $n$, Algorithm~\ref{Algorithm1} outputs  \texttt{"Unidentifiable case"} with probability at least $1 - 2\alpha$.
\end{lemma}

Extending Lemma~\ref{Lemma3} to the non-oracle setting or Proposition~\ref{consistency_proposition} beyond the bivariate case remains technically challenging and is left for future work. Although \cite{Kernel_based_tests} proved the asymptotic validity of bootstrap-based $p$-values for the multivariate HSIC test, estimation error in $\hat{\theta}$ propagates to the residuals and alters the variability of the HSIC statistic, complicating its asymptotic analysis. A rigorous generalization would require deriving asymptotic distributions for higher-order kernel statistics under estimated residuals, which is beyond the current theoretical scope of our work.  Empirically, however, as shown in Section~\ref{Section_simulations_Pareto}, the resulting $p$-values remain well-calibrated even in the presence of estimation error.

\subsubsection{Scalability and greedy algorithms for larger dimensions}
\label{section_scalability}
The main disadvantage of the proposed method is that we have to go through all graphs $\mathcal{G}\in DAG(d)$, which is possible only for very small $d$. However, even though graph learning is typically NP-hard \citep{NP-hard_score_based_causal_learning}, numerous algorithms have been proposed to speed up the process \citep{Greedy_search, Silander2006,  Ramsey2016, Nandy2018, Bregmans_information}.

The \textbf{naive-edge-greedy} algorithm \citep{Greedy_search} is one such algorithm that iteratively adds or removes edges to minimize the CPCM score \eqref{score_definition1}. Another is the \textbf{RESIT} algorithm \citep{Peters2014}, which first estimates a topological ordering and then prunes redundant edges via independence tests. RESIT can be naturally adapted to the CPCM framework (Algorithm~\ref{alg:resit-cpcm} in Appendix~\ref{appendix_greedy_definitions}) and retains large-sample consistency guarantees (Lemma~\ref{thm:resit_consistency} in Appendix~\ref{appendix_greedy_definitions}). However, its performance tends to deteriorate in higher dimensions where early ordering errors and type~I error introduce spurious edges.

%For example, the \textbf{naive-edge-greedy} search algorithm \citep{Greedy_search} begins with an empty DAG and iteratively explores neighboring graphs by adding or removing a single edge. Each candidate DAG is evaluated using the CPCM score function \eqref{score_definition1}, and the graph with the lowest score is selected. If this best candidate improves upon the current graph, it replaces it and the process repeats. The algorithm terminates when no further improvements are possible.

%Another example is \textbf{RESIT} (Regression with Subsequent Independence Test) algorithm, a polynomial-time procedure that first estimates a topological ordering and then prunes superfluous edges via independence tests \citep{Peters2014}.  This method can be naturally adapted for the CPCM framework (Algorithm~\ref{alg:resit} in Appendix~\ref{appendix_greedy_definitions}) and comes with theoretical guarantees for large-sample consistency (Lemma~\ref{thm:resit_consistency} in Appendix~\ref{appendix_greedy_definitions}). Despite its computational efficiency and theoretical backing, RESIT tends to perform poorly in high-dimensional settings. Errors made in the early stages of ordering propagate downstream, and false positives in Phase 2 lead to persistent spurious edges in the final output.

To combine the strengths of both approaches, we introduce a hybrid \textbf{RESIT-greedy} algorithm that merges Phase~1 of RESIT with the edge-pruning strategy of greedy search; see Algorithm~\ref{Algorithm_RESIT_GREEDY}. In Phase~1, the topological ordering is estimated by iteratively removing the node whose residuals $\hat{\varepsilon}_i$ exhibit the weakest dependence on the remaining variables; this node is then appended to the ordering. Alternative procedures for estimating the topological order also exist \citep{Gnecco}. In Phase~2, instead of performing independence tests, we apply a greedy edge-removal procedure guided by the CPCM score: starting from the fully connected graph consistent with the estimated ordering, we iteratively remove the edge whose deletion yields the largest decrease in the CPCM score, stopping when no further improvement is possible.

In Appendix~\ref{appendix_greedy_definitions}, we compare these algorithms in terms of accuracy and computational time. The results (unsurprisingly) show that the exact method achieves the highest accuracy but is computationally feasible only for small graphs ($d <5$). RESIT and naive-greedy methods are the most scalable but exhibit worst accuracy. The RESIT-greedy algorithm strikes a balance: it consistently outperforms both RESIT and naive-greedy in terms of accuracy while maintaining reasonable computational efficiency for moderate dimensions ($d < 10$), making it a practical choice in such cases.

For much larger dimensions, a hybrid approach is recommended: using e.g. PC algorithm to estimate the skeleton and applying the CPCM algorithm only to smaller subgroups of unoriented edges. Similar strategy has been discussed, for example, in \citep{Goudet2017}, though a detailed exploration is beyond the scope of this paper.

\begin{algorithm}[]
\caption{RESIT-greedy algorithm}
\label{Algorithm_RESIT_GREEDY}
\KwIn{Random sample of $(X_1,\dots,X_d)$}
\textbf{Phase 1 (topological order):} Obtain topological order $\pi$ via RESIT (Algorithm~\ref{alg:resit-cpcm}).

\textbf{Phase 2 (edge removal):} Initialize $\mathcal{G} \gets \{(j\!\to\! i): j \text{ precedes } i \text{ in } \pi\}$.

\Repeat{\textbf{no improvement}}{
  $S \gets s(\mathcal{G})$, where $s(\cdot)$ is the score defined in Equation~\eqref{score_definition1} or \eqref{multi_extention}\;
  $e^\star \gets \arg\min_{e \in \mathcal{G}} s(\mathcal{G}\setminus\{e\})$; 
  
  \lIf{$s(\mathcal{G}\setminus\{e^\star\}) < S$}{$\mathcal{G}\gets \mathcal{G}\setminus\{e^\star\}$}
}
\KwOut{$\mathcal{G}$}
\end{algorithm}

\subsection{Choice of the collection $\{F_1, \dots, F_k\}$}
\label{Section5Model_choice}

\subsubsection{Why $\{F_1, \dots, F_k\}$ cannot be chosen in a data-driven way }

Selecting the collection \(\{F_1, \dots, F_k\}\) is a crucial step in our approach. Unfortunately, there is no principled, general data-driven way to select this collection without access to alternative data. As discussed in Section~\ref{motivation_section}, alternative methods such as ANM, QVF, bQCD, or IGCI all predefine the notion of a ``simple'' distribution. This predefinition is necessary; otherwise, Occam's razor loses its meaning. The following lemma formalizes this idea.  

\begin{lemma}\label{lemma_o_overparametrizacii}
Suppose that the joint distribution $F_{(X_1,X_2)}$ is generated according to the model $CPCM(F_2)$ with graph $X_1\to X_2$, where $F_2$ is a distribution function belonging to the exponential family. 

Then, there exists $F_1$ such that the model $CPCM(F_1)$ with graph $X_2\to X_1$ also generates $F_{(X_1,X_2)}$. In other words, there exists $F_1$ such that the causal graph in $CPCM(F_1, F_2)$ is not identifiable from the joint distribution. 
\end{lemma}

The proof (provided in \hyperref[Proof of lemma_o_overparametrizacii]{Appendix} \ref{Proof of lemma_o_overparametrizacii}) is based on the specific choice of $F_{1}$, such that its sufficient statistic is equal to the parameter $\theta_2$ from the original model $CPCM(F_{2})$. 

It is important to note that such an $F_1$ often results in a rather non-standard distribution. While the notion of a ``standard'' distribution is somewhat philosophical, some distributions are objectively considered standard due to physical motivations; for example, the Gaussian (by the central limit theorem) or the Poisson (which arises when counting independent events). This motivates the definition of a ``standard set of well-known distributions,'' which we discuss in Section~\ref{Section_practical_choices}.

\subsubsection{Unfair game issue}

Even if the theory suggests that it is not possible to fit a $CPCM(F_1, \dots, F_k)$ in both causal directions, this is only an asymptotic result and may not hold in practice with finite samples. Overparameterized models may overfit and capture spurious patterns, while underparameterized ones may violate key assumptions. To ensure a fair comparison, we recommend selecting all $F_i$ with the same number of parameters ($q_i = q_j$). For example, choosing $F_1$ with one parameter and $F_2$ with three ($q_1 = 1$, $q_2 = 3$) introduces bias, since the more flexible model is more likely to fit the data regardless of its correctness (unless we are in the asymptotic regime). We refer to this issue as an ``unfair game.'' 

\subsubsection{Practical choices of the set of ``standard well-known distributions''}
\label{Section_practical_choices}

We define nested sets of ``standard set of well-known distributions''; set $\mathscr{S}_1$  containing distributions with one parameter, and a more complex set $\mathscr{S}_2$ containing distributions with two parameters. We do this to avoid the  ``unfair game'' issue. Both sets should be rich enough to contain a wide range of distributions with different supports and characteristics, but should not contain many distributions with the same support in order to avoid unidentifiable setups. 

For practical purposes, we restrict our attention to distributions that are implemented in \texttt{mgcv} package in \texttt{family.mgcv} \citep{Wood}.  These include:
\begin{enumerate}
    \item $\mathscr{S}_1$ consists of the following one parameter distributions: Gaussian with fixed variance, Poisson, Pareto and Exponential distribution. 
    \item  $\mathscr{S}_2$ consists of the following two parameter distributions: Gaussian, Negative binomial, Generalized Pareto and Gamma distribution.
\end{enumerate}
These choices are tactically made such that every distribution in the family $\mathscr{S}_1$ is a special case of some distribution in $\mathscr{S}_2$; for example, the Poisson distribution arises as a special (limiting) case of the Negative Binomial distribution \citep[see][p. 96]{CasellaBerger2024}. 

We emphasize that many other choices are possible, and our collections are neither exhaustive nor immutable. They may be adapted for specific applications where alternative distributions could be regarded as “standard.” 
\subsubsection{$\mathscr{S}_1$ or $\mathscr{S}_2$? Sequential approach}
\label{S1_or_S2}
We propose a sequential approach for the selection between \(\mathscr{S}_1\) and \(\mathscr{S}_2\). We first choose the (least complex) set \(\mathscr{S}_1\); if this choice is ``not appropriate'', we then extend the set of distributions to \(\mathscr{S}_1\cup \mathscr{S}_2\). If $\mathscr{S}_1\cup \mathscr{S}_2$ is also ``not appropriate'', we may further extend the model class by introducing $\mathscr{S}_3$: the class of standard three-parameter distributions. This hierarchical selection process incrementally increases the complexity of the conditional distributions until at least one graph is plausible. 

We define that \(\mathscr{S}_1\) is ``not appropriate'' if there does not exist any \textit{plausible} graph $\mathcal{G} \in DAG(d)$ under $CPCM(\mathscr{S}_1)$. This is formally defined in Algorithm~\ref{Algorithm_sequential}.  

\begin{algorithm}[]
  \SetAlgoLined
    \textbf{Exact version:} For each ${\mathcal{G}} \in \text{DAG}(d)$, infer whether it is plausible under $CPCM(\mathscr{S}_1)$ (as discussed in Section~\ref{Section_score_based_algorithm}). If no graph is plausible, return $\mathscr{S}_1 \cup \mathscr{S}_2$; otherwise, return $\mathscr{S}_1$.\\  

    \textbf{Greedy version:} When employing a (RESIT-)greedy algorithm, scores are evaluated for multiple candidate DAGs. For each candidate, test plausibility under \(CPCM(\mathscr{S}_1)\); if none are plausible, return \(\mathscr{S}_1 \cup \mathscr{S}_2\); otherwise, return \(\mathscr{S}_1\). \textbf{Note:}~The plausibility check is effectively free, as it relies on the same dependence measure $\rho(\hat{\varepsilon}_1, \dots, \hat{\varepsilon}_d)$ already computed in the score.\\
    
    \textbf{Fast version:} Estimate $\hat{\mathcal{G}}$. Test its plausibility under $CPCM(\mathscr{S}_1)$; if $\hat{\mathcal{G}}$ is not plausible return $\mathscr{S}_1\cup \mathscr{S}_2$; otherwise, return~$\mathscr{S}_1$.
  \caption{Sequential approach for the choice between $\mathscr{S}_1$ and $\mathscr{S}_2$}
  \label{Algorithm_sequential}
\end{algorithm}
By sequentially expanding \(\mathscr{S}_1\) by \(\mathscr{S}_2\), and potentially by \(\mathscr{S}_3\), we ensure that the simplest adequate model is chosen while avoiding unnecessary complexity and unfair game issue. As an example when our data are continuous uni-modal with full support, consider that a Gaussian model with fixed variance is fitted in both directions $X\to Y$ and $Y\to X$. If both graphs are not plausible, we expand the model to location-scale Gaussian model. Potentially, if both directions are again not plausible in this model, we add a shape parameter and consider Generalized Gaussian model \citep{Nadarajah01092005}.  

\begin{consequence}
\label{consequence_consistency}
Let $(X_1, X_2)$ follow an $CPCM(F)$ model with DAG $\mathcal{G}$, for some $F\in\mathscr{S}_1\cup \mathscr{S}_2$. Assume the conditions of Proposition~\ref{consistency_proposition} hold: namely, identifiability of $CPCM(\mathscr{S}_1 \cup \mathscr{S}_2)$, employing a ``suitable'' estimation procedure and the HSIC score. Then, our score-based algorithm, with the collection $\{F_1, \dots, F_k\}$ chosen via the Sequential approach (Exact version), is consistent:
$\hat{\mathcal{G}} \overset{P}{\to}\mathcal{G}$, as $n\to\infty$.
\end{consequence}

\section{Experiments} \label{simulations_section}
The \texttt{R} code for the presented algorithms, simulations, and application is available at \url{https://github.com/jurobodik/Causal_CPCM.git}.

\begin{itemize}
   \item In Section~\ref{Section_simulations_robustness}, we investigate the robustness of our approach against misspecifications of the choice of \(F\) (the data are generated with \(F\) being an exponential distribution but we use different choices for $F$ such as Gaussian or Pareto). 
   \item In Section~\ref{Section_simulations_Pareto}, we empirically validate Consequence~\ref{paretoidentifiability} and Proposition~\ref{consistency_proposition}, showing that the distribution of the $p$-values (scores) produced by Algorithm~\ref{Algorithm1} across different causal graphs is approximately uniform.
    \item In Section~\ref{Section_simulations_Gaussian} and Section~\ref{Section_simulations_multivariate}, we compare our methodology with state-of-the-art methods using bivariate and multivariate benchmark datasets, respectively. 
    \item In Section~\ref{Section7} we consider a toy example on real-world data. 
    \item In Appendix~\ref{appendix_greedy_definitions}, we evaluate the performance of different greedy algorithms. This simulation study suggests that RESIT-greedy achieves better accuracy than standard greedy or RESIT methods, but at the cost of slightly higher computational complexity, making it a practical choice for moderately sized graphs ($d<10$). 
    \item In Appendix~\ref{Simulations_seq_approach}, we empirically evaluate the Sequential approach, compared to oracle choice of $F$. 
\end{itemize}

\textbf{Implementation details.}  In all experiments, we estimate ${\theta}(X_i)$ using Generalized Additive Models (GAM, \cite{Wood2}). We use the HSIC as an independence test and approximate the null distribution with a gamma distribution
in order to obtain p-values \citep{Kernel_based_tests}.  In Section~\ref{Section_simulations_Pareto}, we use the conservative estimate in Algorithm~\ref{Algorithm1}; elsewhere, for comparability with other methods, we use the Forced/Score-based estimate (recall that the forced estimate is a special case of the score-based estimate when $d=2$). For the sequential approach, the exact version is applied with the exact CPCM algorithm, the fast version with RESIT, and the greedy version with RESIT-greedy or edge-greedy algorithms.

\subsection{Robustness against a misspecification of F}
\label{Section_simulations_robustness}

We consider the structural model $X_1 \to X_2$, where $X_1 \sim N(2,1)^+$ (Gaussian truncated to $x > 0$) and
\begin{equation}
    \label{asfdae}
    X_2 \mid X_1 \sim \mathrm{Exponential}\big(\theta(X_1)\big),
\end{equation}
for some non-negative function $\theta$. This corresponds to model~(\ref{BCPCM}) with $F$ equal to the exponential distribution, a special case of the Gamma family with a fixed shape parameter.

The simulation evaluates how misspecifying $F$ affects causal graph estimation. We generate $n = 1000$ samples from~\eqref{asfdae}, apply our framework for several candidate families, and record the proportion of correctly identified directions (using Forced Algorithm~\ref{Algorithm1}), averaged over 100 repetitions. 

Results are shown in Table~\ref{table_simulations_about_misspecified_F}. The method is robust when $F$ resembles the exponential distribution in support and tail behavior, but accuracy drops sharply for e.g. $F = \text{Gaussian}$. Surprisingly, $F = \text{Gamma}$, the correctly specified family, performs relatively poorly due to over-parameterization: the reverse model $X_2 \to X_1$ can often fit well with two parameters. Using three-parameter families typically reduces accuracy to about $50\%$, no better than a coin toss, showing the importance of controlling model complexity in causal discovery.

\begin{table}[ht]
\centering
\renewcommand{\arraystretch}{1.15}
\begin{tabular}{l c S S S S}
\toprule
\textbf{Family $F$} & {\#param} &
{$\theta(x) = \mathrm{random}$} &
{$\theta(x) = x$} &
{$\theta(x) = x^2 + 1$} &
{$\theta(x) = e^x / 2$} \\
\midrule
\rowcolor{RowAlt}
Exponential (oracle)                  & 1 &  0.99 &  0.98 & 0.99 &  1.00 \\
Gamma, fixed scale                    & 1 & 0.96 & 0.96 & 1.00 &  0.99 \\
\rowcolor{RowAlt}
Pareto, shifted support                                 & 1 & 0.99 &  0.99 &  1.00 & 1.00 \\
Gumbel, fixed scale                    & 1 & 0.36 & 0.00 & 0.01 & 0.00 \\
\rowcolor{RowAlt}
Gaussian, fixed $\sigma$               & 1 & 0.01 & 0.00 & 0.01 & 0.00 \\
\midrule
Gamma                                  & 2 & 0.96 & 0.73 & 0.64 & 0.79 \\
\rowcolor{RowAlt}
Gumbel                                 & 2 & 0.68 & 0.87 & 0.91 & 0.97 \\
Gaussian                               & 2 & 0.69 & 0.07 & 0.32 & 0.29 \\
\bottomrule
\end{tabular}
\caption{Accuracy of CPCM estimations for different distribution families $F$, with the Exponential distribution as the ground truth. “Random” $\theta(x)$ is generated via Gaussian processes as detailed in  Appendix~\ref{Appendix_Section_simulations_Pareto}.}
\label{table_simulations_about_misspecified_F}
\end{table}

\subsection{Empirical validation of 
 Consequence~\ref{paretoidentifiability}, Proposition~\ref{consistency_proposition} and p-value distribution in Pareto model}
\label{Section_simulations_Pareto}

As in Consequence~\ref{paretoidentifiability}, we consider the CPCM Pareto model $X_1 \to X_2$ with
\begin{equation}\label{fwesef}
    p_{X_1}(x) \propto \frac{1}{[\log(x)+1] x^{2}}, \quad   
X_2 \mid X_1 \sim \mathrm{Pareto}\big(\theta(X_1)\big), \quad  
\theta(x) = x^\gamma \log(x) + 1,
\end{equation}
where $\gamma \in \mathbb{R}$ measures deviation from the unidentifiable case $(\gamma=0)$. For $\gamma > 0$, the causal graph should be identifiable; for $\gamma < 0$, $\theta$ is nearly constant and $X_1, X_2$ are almost independent.

Using $\gamma \in \{-2, -1, 0, 1, 2\}$, we apply the Conservative Algorithm~\ref{Algorithm1} with Pareto $F$ and sample size $n=500$. Note that Algorithm~\ref{Algorithm1} has five possible outcomes: 1) $X_1\indep X_2$, 2) $X_1\to X_2$, 3) $X_2\to X_1$, 4) ``unidentifiable setup'' (both directions appear to be plausible) and 5) ``Assumptions not fulfilled'' (neither direction appears to be plausible). 
 \begin{itemize}
     \item Table~\ref{Pareto_simulations1} in Appendix~\ref{Appendix_Section_simulations_Pareto} shows the results averaged over 100 repetitions. The results align with the theory: if \(\gamma = 0\), we typically estimate both directions to be plausible. If \(\gamma > 0\), we tend to estimate the correct direction \( X_1 \to X_2 \); if \(\gamma < 0\), we tend to estimate an empty graph since \( X_1 \) and \( X_2 \) are (nearly) independent.
     \item Figure~\ref{Pareto_histograms} in Appendix~\ref{Appendix_Section_simulations_Pareto} shows the distributions of the p-values from the independence test in Step 1b) of Algorithm~\ref{Algorithm1}, using data simulated with \(\gamma = 0\) and \(\gamma = 2\). The distribution of the p-values appears roughly uniform on $(0,1)$ in the \(\gamma = 0\) case and in the correct direction \( X \to Y \), while the p-values in the direction \( Y \to X \) are typically very close to $0$ in the identifiable setup  $\gamma=2$.
     \item Figure~\ref{sample_size_pareto} in Appendix~\ref{Appendix_Section_simulations_Pareto} shows an empirical validation of the consistency result from Proposition~\ref{consistency_proposition}. For a range of sample sizes \( n \), we generate the dataset using \(\gamma = 1\) and \(\gamma = 2\) as the hyperparameters, and we compute the percentage (out of 100 repetitions) of correctly estimated causal direction. The algorithm appears to achieve near-perfect performance for large sample sizes. 
 \end{itemize}

\subsection{Comparison with baseline methods: bivariate case}
\label{Section_simulations_Gaussian}
We compare our method with LOCI \citep{immer2022identifiability}, HECI  \citep{xu2022inferring}, RESIT \citep{Peters2014},  bQCD \citep{Natasa_Tagasovska}, IGCI with Gaussian and uniform reference measures \citep{IGCI} and Slope \citep{Slope}. As in \cite{reviewANMMooij}, we use the accuracy for forced decisions as our evaluation metric. Details can be found in  Appendix~\ref{Appendix_Section_simulations_Pareto}.

We consider seven benchmark datasets, described below. The first five datasets are taken directly from \cite{Natasa_Tagasovska} and described in Appendix~\ref{Appendix_Section_simulations_Pareto}. They consist of additive and location-scale Gaussian pairs of the form $X_2 = \mu(X_1) + \sigma(X_1) \varepsilon_2$, where $\varepsilon_2 \sim \mathcal{N}(0, 1)$. In each case, we generate $X_1 \sim N(0, \sqrt{2})$. 

\begin{enumerate}
    \item \textbf{LSg (Location-Scale Gaussian)}: Here, $\mu$ and $\sigma$ are nonlinear functions simulated using Gaussian processes.
    \item \textbf{LSs (Location-Scale Sigmoid)}: In this setup, $\mu$ and $\sigma$ are sigmoid functions.
    \item \textbf{ANMg (Additive Noise Model)}: Nonlinear additive noise models generated similarly to LSg, but with constant $\sigma(X_1) = \sigma \sim U(1/5, \sqrt{2/5})$.
    \item \textbf{ANMs (Additive Noise Model)}: Nonlinear additive noise models generated similarly to LSs, but with constant $\sigma(X_1) = \sigma \sim U(1/5, \sqrt{2/5})$.
    \item \textbf{MNs (Multiplicative Noise)}: Nonlinear multiplicative noise models generated similarly to LSs, but with $\mu(X_1) = 0$.
\end{enumerate}

In addition to the models from \cite{Natasa_Tagasovska}, we consider two more setups:

\begin{enumerate}
    \setcounter{enumi}{5}
    \item \textbf{POISg (Poisson Model)}: $X_2 \sim \text{Pois}\big(\lambda(X_1)\big)$, where $\lambda$ is generated using Gaussian processes similar to $\sigma$ in LSg. Observe that $X_2$ is discrete, creating error in some methods. 
    \item \textbf{PARg (Pareto Model)}: $X_2 \sim \text{Pareto}\big(\theta(X_1)\big)$, where $\theta$ is generated using Gaussian processes similar to $\sigma$ in LSg.
\end{enumerate}

For each of the seven setups, we simulate 100 pairs with $n=1000$ data points each. One realization of each model can be found in Figure \ref{Simulations2_plots} in Appendix~\ref{Appendix_simulations}. 

The results are presented in Table \ref{Table_Simulated_data_Gaussian}. We conclude that our estimator performs well across all considered datasets, effectively handling the mix of discrete, continuous and heavy-tailed variables. Under the Gaussian location-scale setups, it provides comparable results to LOCI and bQCD, which are specifically developed for such cases.
% --- Preamble additions (if not already present) ---
% \usepackage{booktabs}
% \usepackage{siunitx}
% \usepackage{xcolor,colortbl}
% \usepackage{makecell}
% \definecolor{RowAlt}{gray}{0.97}
% \sisetup{
%   table-number-alignment = center,
%   table-format=3.0,
%   detect-weight = true,
%   detect-family = true
% }

% \usepackage{booktabs}
% \usepackage{siunitx}
% \usepackage{xcolor,colortbl}
% \usepackage{makecell}
% \definecolor{RowAlt}{gray}{0.97}
% \sisetup{
%   table-number-alignment=center,
%   table-text-alignment=center,
%   table-format=3.0,
%   detect-weight=true,
%   detect-family=true
% }

\begin{table}[!ht]
\centering
\renewcommand{\arraystretch}{1.15}
\begin{tabular}{l
                S[table-format=3.0]
                S[table-format=3.0]
                S[table-format=3.0]
                S[table-format=3.0]
                S[table-format=3.0]
                S[table-format=3.0]
                S[table-format=3.0]}
\toprule
\textbf{Method} &
\textbf{ANMg} & \textbf{ANMs} & \textbf{MNs} &
\textbf{LSg} & \textbf{LSs} & \textbf{POISg} & \textbf{PARg} \\
\midrule
\rowcolor{RowAlt}
\makecell[l]{\textbf{CPCM} \footnotesize(Seq. choice)\\ \footnotesize Forced Algorithm 1}
  & \bfseries 100 & \bfseries 99 & \bfseries 88 & \bfseries 98 & \bfseries 92 & \bfseries 94 & \bfseries 90 \\
\textbf{LOCI} \footnotesize (NN H)
  & \bfseries 100 & \bfseries 100 & \bfseries 99 & \bfseries 91 & \bfseries 85 & 79 & 0 \\
\rowcolor{RowAlt}
\textbf{HECI}
  & \bfseries 98 & \bfseries 43 & \bfseries 29 & \bfseries 96 & \bfseries 54 & \multicolumn{1}{c}{\textemdash} & 100 \\
\textbf{ANM-RESIT}
  & \bfseries 100 & \bfseries 100 & 39 & 51 & 11 & 0 & 12 \\
\rowcolor{RowAlt}
\textbf{bQCD} \footnotesize(m=3)
  & \bfseries 100 & \bfseries 79 & \bfseries 99 & \bfseries 100 & \bfseries 98 & 97 & 34 \\
\textbf{IGCI} \footnotesize(Gauss)
  & 100 & 99 & 99 & 97 & 100 & 0 & 0 \\
\rowcolor{RowAlt}
\textbf{IGCI} \footnotesize(Unif)
  & 31 & 35 & 12 & 36 & 28 & 0 & 100 \\
\textbf{Slope}
  & 22 & 25 & 9 & 12 & 15 & 0 & 100 \\
\bottomrule
\end{tabular}
\caption{Accuracy (\%) of different estimators on the simulated datasets. Similar results (excluding the first row and the last two columns) can also be found in \cite{immer2022identifiability}. Bold entries indicate cases where high accuracy is expected because the data-generating mechanism aligns with the method’s modeling assumptions.}
\label{Table_Simulated_data_Gaussian}
\end{table}

\subsection{Comparison with baseline methods: multivariate case}
\label{Section_simulations_multivariate}
We compare our proposed CPCM method with several widely used causal discovery algorithms implemented in \texttt{R}. These include LINGAM \citep{Lingam}, ANM (RESIT) \citep{Peters2014,hoyer2009}, PC (constraint-based method that tests for conditional independences, \cite{pcalg_package}), and GES (score-based algorithm that greedily searches over equivalence classes of DAGs using the BIC score, \cite{Ramsey2016}). We also include a random baseline that selects a DAG uniformly at random from the space of all DAGs on $d$ nodes. Many other methods, such as NOTEARS \citep{zheng2018dags} or DAG-GNN \citep{yu2019dag}, would also be appropriate comparisons. However, we restrict our evaluation to methods directly available in \texttt{R} to ensure consistency and reproducibility within our implementation framework. More details can be found in Section~\ref{Appendix_Section_simulations_Pareto}. 

We generate data from 4 different scenarios:
\begin{itemize}
    \item \textbf{Linear (non-gaussian)}: $X_{i}=\sum_{k\in pa(i)}X_k + \varepsilon_k$, where $\varepsilon_k\overset{i.i.d}{\sim} Exp(1)$,
    \item \textbf{Nonlinear ANM}: $X_{i}\sim N(\mu(\textbf{X}_{pa_i}), 1)$ where $\mu(\textbf{X}_{pa_i}) = \sum_{k\in pa(i)}sin(X_k) + \frac{1}{2}X_k + \varepsilon_k$, where $\varepsilon_k\overset{i.i.d}{\sim} N(0,1)$
     \item \textbf{CPCM(Exp)}: $X_{i}\sim Exp(\lambda(\textbf{X}_{pa_i}))$, where $\lambda(\textbf{X}_{pa_i}) = \sum_{k\in pa(i)}|X_k|\vee 0.1$,
     \item \textbf{CPCM(Exp, Gauss)}: each node is, with probability $\frac{1}{2}$,  generated either according to the nonlinear ANM case or the $CPCM(Exp)$ case. 
\end{itemize}

The underlying DAG $\mathcal{G}$ is generated uniformly at random by sampling a random ordering of the $d$ variables and including each of the $\tfrac{d(d-1)}{2}$ admissible edges with probability $\tfrac{2}{d-1}$, resulting in an expected total of $d$ edges (following \cite{Peters2014}). We report results for $d=5$; results for $d=10$ are similar and omitted for brevity. For each of the four scenarios, we simulate 50 graphs with $n = 1000$ data points each. We compare the methods using the Structural Intervention Distance (SID, \cite{peters2014structuralinterventiondistancesid}). For fairness, undirected edges in the estimated graph are counted as correct if the true edge exists in either direction.

In the linear case, it performs nearly as well as LINGAM, which is tailored to linear non-Gaussian models. In the nonlinear ANM scenario, it is only slightly less accurate than ANM \citep{Peters2014}. The comparison between RESIT and its greedy variant highlights that in lower dimensions, RESIT-greedy tends to yield better accuracy. Most importantly, in the non-Gaussian $CPCM(F)$ settings, our method clearly outperforms all baselines, underscoring its robustness and suitability for non-additive and heterogeneous functional forms, where standard approaches often fail.

\begin{table}[ht]
\centering
\renewcommand{\arraystretch}{1.15}
\begin{tabular}{>{\raggedright\arraybackslash}p{5.7cm}|
                S[table-format=2.2]
                S[table-format=2.2]
                S[table-format=2.2]
                S[table-format=2.2]}
\toprule
\textbf{Method / Scenario} \footnotesize(Case $d=5$) &
\multicolumn{1}{c}{\makecell{\textbf{Linear}\\\textbf{model}}} &
\multicolumn{1}{c}{\makecell{\textbf{Nonlin.}\\\textbf{ANM}}} &
\multicolumn{1}{c}{\makecell{\boldmath$CPCM(F_1)$\\ \footnotesize$F_1=\mathrm{Exp}$}} &
\multicolumn{1}{c}{\makecell{\boldmath$CPCM(F_1,F_2)$\\ \footnotesize$F_1=\mathrm{Exp},\,F_2=\mathrm{Gauss}$}} \\
\midrule
\rowcolor{RowAlt}
\textbf{CPCM} \footnotesize(Seq. app., RESIT-greedy)                                  & 0.22 & 0.72 & \bfseries 2.16 & \bfseries 0.92 \\
\textbf{LINGAM}                                & \bfseries 0.12 & 4.72 & 4.52 & 5.51 \\
\rowcolor{RowAlt}
\textbf{ANM} \footnotesize(RESIT)              & 1.10 & 1.34 & 12.21 & 6.21 \\
\textbf{ANM} \footnotesize(RESIT-greedy) & \bfseries 0.12 & \bfseries 0.66 & 11.81 & 6.84 \\
\rowcolor{RowAlt}
\textbf{PC} \footnotesize(gaussCItest)         & 1.90 & 2.20 & 4.28 & 3.15 \\
\textbf{PC} \footnotesize(HSIC)         & \bfseries 0.12 & 1.50 &  2.32 & 1.92   \\ \rowcolor{RowAlt}
\textbf{GES} \footnotesize(Gaussian score)     & 1.96 & 2.34 & 4.24 & 3.15 \\
\textbf{Random}                                & 7.20 & 7.58 & 7.50 & 7.57 \\
\bottomrule
\end{tabular}
\caption{Average SID over 100 repetitions between estimated and true graphs under different methods and scenarios. Bold values are the best (lowest) across methods.}
\label{tab:sid_summary}
\end{table}

\subsection{Illustration using a motor insurance dataset}
\label{Section7}

Understanding the causal relationships between driver and vehicle characteristics is a key step in insurance analytics, as it informs both risk assessment and premium setting.
We demonstrate the advantages of CPCM using a subset of the French MTPL motor insurance dataset \citep{sarpa2025freMTPL2}, restricted to four variables:  
``ClaimNb'' records the number of claims during the policy period (integer between 1 and 16) and lending itself naturally to a Poisson model.  
``VehPower'' and ``VehAge'' are vehicles engine power and age, and ``Exposure'' is the duration that the policy was active.  
Exponential/Gamma distribution is a natural model for variables ``VehPower'', ``VehAge'', and ``Exposure'' due to their exponential-type shapes. Since the dataset contains almost a million observations, we focus on a subset of the first $n = 1000$ records; results based on different random subsamples are provided in Appendix~\ref{Appendix_Section_simulations_Pareto}.

Fitting an LINGAM/ANM to this dataset is problematic for two reasons.  
First, the discrete nature of ``ClaimNb'' violates the continuous additive noise assumption used by most ANM methods \citep{Peters_discrete}.  
Second, the other variables are skewed, non-Gaussian and heteroskedastic, making  CPCM much more natural choice.

Using the sequential family-selection approach within CPCM and the RESIT-greedy algorithm, we obtain the graph shown in Figure~\ref{Fig_motor}. In this case, family $\mathscr{S}_1$ was selected, as the resulting graph was marked as {plausible}. In contrast, applying LiNGAM or ANM-RESIT yields markedly different structures: LiNGAM suggests $\mathrm{Exposure} \rightarrow \mathrm{VehAge} \rightarrow \mathrm{VehPower}$, while ANM-RESIT recovers only $\mathrm{VehAge} \rightarrow \mathrm{VehPower}$.  

Although the ground truth for this dataset is not perfectly clear, the results provide evidence that CPCM can recover more plausible causal structures in mixed-type insurance data than existing ANM-based or linear approaches.  
By accommodating discrete outcomes, non-Gaussian noise, heteroskedasticity and heavy-tails, CPCM produces graphs that align better with domain knowledge and avoid the misspecifications that can arise in more restrictive frameworks.  
On the other hand, CPCM can be computationally demanding for larger datasets, especially when sequential family selection is combined with many candidate parent sets.  
Furthermore, as with most observational causal discovery methods, CPCM relies on the assumption of causal sufficiency, which can be restrictive in practical applications.

\begin{figure}[h]
\centering
\begin{tikzpicture}[
    every node/.style={draw, rectangle, rounded corners, minimum width=1.7cm, minimum height=0.9cm, align=center},
    >=stealth,
    node distance=1.2cm % reduced vertical gap
]

% Nodes
\node (VehAge) {VehAge};
\node[right=of VehAge] (VehPower) {VehPower};
\node[right=of VehPower] (Exposure) {Exposure};
\node[below=0.9cm of VehPower] (ClaimNb) {ClaimNb};

% Edges
\draw[->] (VehAge) -- (VehPower);
\draw[->] (VehAge.north east) .. controls +(1.0,0.8) and +(-1.0,0.8) .. (Exposure.north west);
\draw[->] (VehPower) -- (Exposure);
\draw[->] (VehPower) -- (ClaimNb);
\draw[->] (Exposure) -- (ClaimNb);

\end{tikzpicture}
\caption{Causal graph estimated by CPCM on the French motor insurance dataset subset.}
\label{Fig_motor}
\end{figure}

\section{Conclusion and future research}

We introduced a new family of models for causal inference called Conditionally Parametric Causal Models (CPCM), designed to flexibly accommodate a broad range of variable types and distributional forms. Our primary theoretical contributions lie in establishing the identifiability conditions for the causal structure within this framework. Specifically, we have demonstrated that the bivariate $CPCM(F)$ models are identifiable, with exceptions arising only when the parameters of $F$ take the form of a linear combination of its sufficient statistics. Furthermore, we have provided detailed characterizations of identifiability across various cases such as Gaussian, Poisson, and Pareto, significantly broadening the scope of identifiable models beyond existing literature. We also explained the multivariate extensions of these results. 

We complement these results with two consistent estimation algorithms for CPCM-based causal graph recovery. Experiments show competitive performance in Gaussian location–scale models, while retaining the ability to operate in much broader distributional settings, including heavy-tailed, continuous, discrete or even a mixture of these. 

CPCM also connects naturally to invariant causal prediction (\cite{Peters_invariance, Kook03042025}), offering promising directions for distribution-aware causal feature selection, as the framework of target-variable causal modeling provides a natural environment for embedding the CPCM ideology \citep{Bodik_biometrika}. Extensions to time series settings \citep{Bodik2024, Assaad_time_series} or uncertainty quantification under distribution shift \citep{liu2021learning, bodik2025crossworldassumptionrefiningprediction} would further broaden its applicability. Integrating these ideas and validating CPCM in diverse applied domains \citep{Gamella2025CausalChambers} are key avenues for future work.

\section*{Conflict of interest and data availability}
The open-source implementation of the methods discussed in this manuscript together with the data used can be found in the supplementary package or at \url{https://github.com/jurobodik/Causal_CPCM.git}.

The authors declare that they have no known competing financial interests or personal relationships that could have appeared to influence the work reported in this paper.

\section*{Acknowledgments}
This study was supported by the Swiss National Science Foundation under grant number 201126.

% ----------------------------------------------------------
% ACKS (if any) – JMLR wants them before references
% ----------------------------------------------------------
% \acks{This research was supported by ...}

% ----------------------------------------------------------
% APPENDIX
% ----------------------------------------------------------

\begin{center}
\Large \textbf{Appendix}
\end{center}
This appendix is organized as follows:
\begin{itemize}
    \item Appendix~\ref{Appendix_A} provides a detailed definition and assumptions concerning the Exponential family, which were omitted from the main text for clarity. It also introduces the notion of F-suitability of an estimator and includes an in-depth discussion and proof of Proposition~\ref{consistency_proposition}.
    \item Appendix~\ref{Appendix_simulations} presents additional experiments and elaborates on certain implementation details.
    \item Appendix~\ref{SectionProofs} contains all theoretical results along with their corresponding proofs.
\end{itemize}

\renewcommand\thesection{A}
\section{Exponential family and Proposition \ref{consistency_proposition}}
\label{Appendix_A}

\subsection{Exponential family}
\label{appendix_exponential_family}

The exponential family is a set of probability distributions whose probability density function can be expressed in the following form:
\begin{equation}\tag{\ref{Exponential family of distributions}}
f(x;\theta) = h_1(x)h_2(\theta)\exp\big[\sum_{i=1}^q\theta_iT_i(x)\big],
\end{equation}
where $h_1, T_i$ are real functions and $h_2:\mathbb{R}^q\to\mathbb{R}^+$ is a vector-valued function. We call $T_i$ a \textit{sufficient} statistic, $h_1$ a base measure, and $h_2$ a normalizing (or partition) function.  

Often, form (\ref{Exponential family of distributions}) is called a canonical form and $$f(x;\theta) = h_1(x)h_2(\theta)\exp\big[\sum_{i=1}^qh_{3,i}(\theta)T_i(x)\big],$$ where  $h_{3,i}:\mathbb{R}^q\to\mathbb{R}, i=1, \dots, q$, is called its \textit{reparametrization} (natural parameters are a specific form of the reparametrization). We always work only with a canonical form (attention for Gaussian distribution, where the standard form is not in the canonical form). 

Numerous important distributions lie in the exponential family of distributions, such as Gaussian, log-normal, Poisson, Pareto (with fixed support), Weibull, chi-squared, multinomial, Binomial, Gamma, and Beta distributions, to name a few. 

It is important to note that functions in (\ref{Exponential family of distributions}) are \textit{not} uniquely defined. For example, $T_i$ is unique up to a linear transformation. 

The support of $f$ is fixed and does not depend on $\theta$. Potentially, $T_i$ and $h_1$ do not have to be defined outside of this support; however, we typically overlook this fact (or possibly define $h_1(x) = T_i(x) = 0$ for $x$ where these functions are not defined). We additionally assume that the support is nontrivial in the sense that it contains at least two distinct values.

Without loss of generality, we assume that $q$ is minimal in the sense that $f(x; \theta)$ cannot be expressed using only $q - 1$ parameters. The sufficient statistics $T_1, \dots, T_q$ are then linearly independent in the following sense: there exist points $x_1, \dots, x_q \in \operatorname{supp}(f)$ such that the matrix 
\begin{equation}\label{eq2431087}
\begin{pmatrix}
T_{1}(x_1) & \cdots & T_{{q}}(x_1) \\
\cdots & \ddots & \cdots \\
T_{1}(x_q) & \cdots & T_{{q}}(x_{q}) 
\end{pmatrix} 
\end{equation}
has full rank. Moreover, $T_1, \dots, T_q$ are affinly independent in the following sense: there exist $y_0, y_1, \dots, y_q\in supp(f)$, such that a matrix 
\begin{equation}\label{eq145151}
\begin{pmatrix}
T_{1}(y_1) - T_1(y_0) & \cdots & T_{{q}}(y_1) -T_q(y_0)\\
\cdots & \ddots & \cdots \\
T_{1}(y_q) -T_1(y_0) & \cdots & T_{{q}}(y_{q}) -T_q(y_0)
\end{pmatrix} 
\end{equation}
has full rank. In this paper (Lemma~\ref{PomocnaLemma1}) we assume affine independence of $T_1, \dots, T_q$, i.e., that condition~\eqref{eq145151} holds.

Since the notions of linear and affine independence used here are nonstandard, we illustrate them with a simple example. Let $T_1(x) = x$ and $T_2(x) = x^2$, corresponding to the sufficient statistics of a Gaussian distribution. Then matrices \eqref{eq2431087} and \eqref{eq145151} become: 
\begin{equation*}
M_1=\begin{pmatrix}
x_1 &  x_1^2\\
x_2  & x_{2}^2
\end{pmatrix} , \,\,\,\,
M_2=\begin{pmatrix}
y_1-y_0 &  y_1^2 -y_0^2\\
y_2 -y_0 &  y_2^2 -y_0^2
\end{pmatrix} ,
\end{equation*}
both of which are full-rank for the choices $(x_1, x_2) = (1,2)$ and $(y_0, y_1, y_2) = (0,1,2)$, for instance. The same analogously holds for all distributions considered in this paper. 

\subsection{F-suitability: Proposition~\ref{consistency_proposition}, Consequence~\ref{consequence_consistency} and Lemma~\ref{Lemma3} }
\label{Appendix_consistency}

\subsubsection{Suitability of an estimator}
In the following, we define an $F$-suitable estimator for a given distribution function $F$ with parameters $\theta$. This is a modification of the concept of a suitable estimator for the conditional expectation discussed in \cite[Appendix A.2]{reviewANMMooij}. In case when $F$ is location-distribution (such as Gaussian distribution with fixed variance), our results fully align with \citep{reviewANMMooij}.

Let $(X_{i}, Y_{i})_{i=1}^n$ be a random sample from $(X, Y)$. We say that the estimator $\hat{\theta}$ is \textbf{F-suitable for $X \to Y$} if the following conditions are satisfied:

\begin{itemize}
    \item \textbf{(Existence of a point-wise limit)} There exists $\theta$ such that $\hat{\theta}(x) \overset{P}{\to} \theta(x)$ as $n\to\infty$ for all $x\in supp(X)$. Moreover, if it is possible to write  $Y = F^{-1}(\varepsilon_2; \tilde{\theta}(X))$, with $X\indep \varepsilon_2\sim Unif(0,1)$, then this limit is equal to $\theta= \tilde{\theta}$. 
    \item \textbf{(Weak residual consistency)} It holds that
\begin{equation}
    \label{gdert}
    \lim_{n\to\infty}\mathbb{E}_{}\left( \frac{1}{n}\sum_{i=1}^n(\varepsilon_i-\hat{\varepsilon}_i)^2 \right) = 0,
\end{equation}
where $\varepsilon_i := F(Y_{i}; \theta(X_{i}))$ and $\hat{\varepsilon}_i := F(Y_{i}; \hat{\theta}(X_{i})), i=1, \dots, n$ and where the expectation is taken with respect to the distribution of the random sample.
\end{itemize}

We say that the estimator $\hat{\theta}$ is \textbf{F-suitable} if it is F-suitable for both $X \to Y$ and $Y \to X$. We simply write that  $\hat{\theta}$ is ``suitable'' if $F$ is evident from the context. 

\subsubsection{Literature review - which estimators are suitable?}

\textbf{Location family:} when $F$ is location-family distribution, such as a Gaussian distribution with fixed variance, property \eqref{gdert} reduces to classical notion of a weak universal consistency:
\(
\lim_{n\to\infty}\mathbb{E}[||\theta(X) - \hat{\theta}(X)||^2] = 0.
\) Such consistency have been already discussed also in relation to causal discovery \citep{zhang2015estimation,uemura2022multivariate, keropyan2023rank}. 
The weak universal consistency has been established for various estimators under appropriate smoothness assumptions on $\theta$. Examples of such estimators include:
\begin{itemize}
    \item Kernel estimators (see Theorem 5.1 in \cite{Gyorfi2002})
    \item Smoothing spline GAM estimators (see Chapter 14.2 in \cite{Gyorfi2002}, \cite{claeskens2009} or \cite{Wood2})
    \item Neural networks (see Theorem 16.1 in \cite{Gyorfi2002} or \cite{drews2022universalconsistencyoverparametrizeddeep, Heiss2024}).
\end{itemize}
Importantly, these consistency results apply regardless of the causal direction. In the anti-causal direction, we have \( X = \theta(Y) + \varepsilon \), where \( \theta(Y) = \mathbb{E}[X \mid Y] \) and \( \varepsilon \not\indep Y \), \( \mathbb{E}[\varepsilon \mid Y] = 0 \). Note that  \( \varepsilon \not\indep Y \) holds if and only if the model is identifiable. For such case, the same form of weak consistency remains valid for the estimators listed above (again, under appropriate smoothness assumptions).

\textbf{Location-scale family:} when $F$ is a location-scale distribution (e.g. Gaussian), several consistency results has also been established for various estimators under appropriate smoothness and regularity assumptions, see e.g. \cite{10.1093/biomet/85.3.645}, or \cite{immer2022identifiability, Siegfried02102023, Le_Smola}. 

\textbf{More general families:} Consistency results for non-Gaussian families are less explored in nonparametric settings, with most existing literature focusing on empirical evidence. GAMLSS \citep{GAMLSS} offers a broad class of estimators, with \cite{GAMLSS_webpage} presenting extensive empirical evidence on simulations and hundreds of real-data examples demonstrating consistency. A few theoretical results are avaliable for GAM estimators; e.g. Theorem 1 in \cite{GAMSamworthConsistency} establishes its almost sure universal consistency for a general one-dimensional exponential family of distributions for \( Y \mid X \), assuming certain smoothness and convexity/monotonicity conditions on \( \theta \). \cite{Wood2} discusses general framework for smoothing parameter estimation for models with regular likelihoods constructed in terms of unknown smooth functions of covariates.   \cite{mammen1997} discusses the consistency under partial linearity assumption.

\subsubsection{HSIC score}

For the definition and discussion of the HSIC score, see \cite[Appendix A.1]{reviewANMMooij}. We use the same notation and implicitly use bounded non-negative Lipschitz-continuous kernels such that their product is characteristic, as in the ``Data recycling'' scenario in \cite[Corollary 21]{reviewANMMooij}.

\subsubsection{Proofs of Proposition~\ref{consistency_proposition}, Consequence~\ref{consequence_consistency} and Lemma~\ref{Lemma3}}

\begin{customprop}{\ref{consistency_proposition}}
Let $(X_1, X_2)$ follow an identifiable $CPCM(F_1, \dots, F_k)$ with DAG $\mathcal{G}$. Then, our score-based algorithm presented in Section~\ref{Section_score_based_algorithm} is consistent, meaning that
$$\hat{\mathcal{G}} \overset{P}{\to}\mathcal{G}\,\,\,as\,\,n\to\infty,$$
given that we employ a “suitable” estimation procedure for the estimator  $\hat{\varepsilon}_i$, we use HSIC score as our choice of $\rho$ and consistent estimates of $S_i$. 
\end{customprop}

\begin{proof}
The proof mostly aligns with the proof of Corollary 21 in \cite{reviewANMMooij}. We use the notation $X = X_1$ and $Y = X_2$.

If $X \indep Y$, then $\hat{\mathcal{G}}$ converges to an empty graph, since any other graph has a score of at least $\lambda$ and $\hat{\varepsilon}_1, \dots, \hat{\varepsilon}_d$ are independent by definition. From now on, without loss of generality, let $Y = F_j^{-1}(\varepsilon_2; \theta(X)), \varepsilon_2\indep X$ for some  $j \in \{1, \dots, k\}$. Denote by $\mathbf{x} = (x_1, \dots, x_n)$ and $\mathbf{y} = (y_1, \dots, y_n)$ the observed data. In the following, we compare the asymptotic scores of graphs $X\to Y$ and $Y\to X$. 

\textbf{Graph} $X \to Y$: Define the ``population residual'' \(E_Y := F_j(Y; \theta(X))\) and the ``estimated residual'' $\hat{\varepsilon}_Y^i = (F_i(y_l; \hat{\theta}(x_l)))_{l=1}^n$, where $i=1, \dots, k$. For the true $i = j$, we omit the superscript and write $\hat{\varepsilon}_Y = (F_j(y_l; \hat{\theta}(x_l)))_{l=1}^n$. By construction, \(X \perp\!\!\!\perp E_Y\) which implies \(\text{HSIC}(X, E_Y) = 0\) (due to Lemma 12 in \cite{reviewANMMooij}). 

Now, since the estimator $\hat{\theta}$ satisfies \eqref{gdert}, we can use the argument presented in \cite[Theorem 20]{reviewANMMooij}, and obtain $\widehat{HSIC}(\mathbf{x}, \hat{\varepsilon}_y) \xrightarrow{P} \text{HSIC}(X, E_Y)$. 

Since $\hat{S}_2$ is consistent, we can find $n_0$ such that for all $n\geq n_0$ holds $P(j\in\hat{S}_2)\geq 1-\delta$ for given $\delta>0$. 

Putting everything together, with probability larger than $1-\delta$ holds

\begin{equation*}
    \begin{split}
        s(X \to Y) &= \min_{j_2 \in \hat{S}_2}\rho(\mathbf{x}, \hat{\varepsilon}_y^{j_2}) + \lambda \leq \rho(\mathbf{x}, \hat{\varepsilon}_y) + \lambda \\
        &= \widehat{HSIC}(\mathbf{x}, \hat{\varepsilon}_y) + \lambda \xrightarrow{P} \text{HSIC}(X, E_Y) + \lambda = \lambda.
    \end{split}
\end{equation*}
By sending $\delta\to 0$, we obtain $  s(X \to Y)\xrightarrow{P} \lambda $. 

\textbf{Graph} $Y \to X$: Define the ``population residual'' \(E_X^j := F_j(X; \theta_j(Y))\) and the ``estimated residual'' $\hat{\varepsilon}_X^j = (F_j(x_i; \hat{\theta}_j(y_i)))_{i=1}^n$, where $\theta_j$ is the limit of $\hat{\theta}_j$ defined in the definition of $F_j$-suitability. 

Due to the assumption of identifiability, $Y\not\indep E_X^j$ for all $j\in\{1, \dots, k\}$. Therefore, using Lemma 12 in \cite{reviewANMMooij}, we have \(\text{HSIC}(Y, E_X^j) > 0\). Again, since the estimator $\hat{\theta}$ satisfies \eqref{gdert}, we can use the argument presented in \cite[Theorem 20]{reviewANMMooij}, and obtain $\widehat{HSIC}(\mathbf{y}, \hat{\varepsilon}_X^j) \xrightarrow{P} \text{HSIC}(Y, E_X^j)$.

Putting everything together
\begin{equation*}
    \begin{split}
        s(Y \to X) &= \min_{j \in \hat{S}_1}\rho(\mathbf{y}, \hat{\varepsilon}_x^j) + \lambda \geq \min_{j \in \{1, \dots, k\}}\widehat{HSIC}(\mathbf{y}, \hat{\varepsilon}_X^j) + \lambda \\
        &  \xrightarrow{P} \min_{j \in \{1, \dots, k\}}\text{HSIC}(Y, E_X^j) + \lambda > \lambda, \,\,\,\,\,\,\,\,\,\,\,\text{as}\,\,n\to\infty.
    \end{split}
\end{equation*}
Therefore, the score for the correct direction is asymptotically smaller than that for the wrong causal direction as $n \to \infty$, hence the procedure is consistent.
\end{proof}

\begin{customconsequence}{\ref{consequence_consistency}}
Let $(X_1, X_2)$ follow an $CPCM(F)$ model with DAG $\mathcal{G}$, for some $F\in\mathscr{S}_1\cup \mathscr{S}_2$. Assume the conditions of Proposition~\ref{consistency_proposition} hold: namely, identifiability of $CPCM(\mathscr{S}_1 \cup \mathscr{S}_2)$, a suitable estimation procedure, and the use of the HSIC score. Then, our score-based algorithm, with the collection $\{F_1, \dots, F_k\}$ chosen via the Sequential approach (Exact or Fast
version), is consistent:
$\hat{\mathcal{G}} \overset{P}{\to}\mathcal{G}$, as $n\to\infty$.
\end{customconsequence}

\begin{proof}
If $F \in \mathscr{S}_1$, the result follows directly from Proposition~\ref{consistency_proposition}. Similarly, if $F \in \mathscr{S}_2$ and the Sequential approach returns $\mathscr{S}_1 \cup \mathscr{S}_2$, the proposition again directly applies.

It remains to consider the case where $F \in \mathscr{S}_2$ but the Sequential approach returns only $\mathscr{S}_1$. We will show that, as $n \to \infty$, this event occurs with probability tending to zero. In other words, when $F \in \mathscr{S}_2$, the Sequential approach will return $\mathscr{S}_1 \cup \mathscr{S}_2$ with probability tending to one.

If $X \indep Y$, then $\hat{\mathcal{G}}$ converges to an empty graph regardless of the choices of $F$. Hence, without loss of generality, assume non-empty causal graph. Consider graph $\mathcal{G} = X_1 \to X_2$, and take some $F_1 \in \mathscr{S}_1$ (that is, wrong distribution $F_1\neq F$). In this case, following the notation of the proof of Proposition~\ref{consistency_proposition}, we have
\begin{equation}
\label{eq23451111}
    \widehat{\text{HSIC}}(\mathbf{x}_1, \hat{\varepsilon}_{X_2}^{1}) \xrightarrow{P} \text{HSIC}(X_1, E_{X_2}^{1}) > 0,
\end{equation}
where $\hat{\varepsilon}_{X_2}^{j} := (F_j\big(X_2; \hat{\theta}_j(X_1)\big))_{l=1}^n$ is the vector of residuals obtained by applying a ``suitable'' estimation procedure, and the ``population residual'' \(E_{X_2}^{j} := F_j(X_2; \theta(X_1))\). This convergence follows from the same reasoning as in the `$Y \to X$' case in the proof of Proposition~\ref{consistency_proposition}. 

In particular, \eqref{eq23451111} implies that the variance of the empirical HSIC statistic vanishes:
\[
\operatorname{Var} \left( \widehat{\text{HSIC}}(\mathbf{x}_1, \hat{\varepsilon}_{X_2}^{1}) \right) \to 0 \quad \text{as } n \to \infty.
\]
Therefore, according to the construction of confidence intervals for the HSIC test in~\cite[Section~3.2.2]{Kernel_based_tests}, the $(1 - \alpha)$ confidence interval will eventually lie entirely within an interval of the form $[\mathrm{HSIC}(X_1, E_{X_2}^{1}) - \delta, \mathrm{HSIC}(X_1, E_{X_2}^{1}) + \delta]$, for arbitrarily small $\delta > 0$. Since the limiting HSIC value is strictly positive, this interval will exclude $0$ for large enough $n$, yielding p-values smaller than $\alpha$. As a result, the DAG $\mathcal{G}$ will be rejected as plausible.

The same argument applies to the case $\mathcal{G} = X_2 \to X_1$ and for any $F_j \in \mathscr{S}_1$, $j = 1, \dots, k$. Therefore, for sufficiently large $n$, no DAG will be deemed plausible, and the Sequential approach will return the full set $\mathscr{S}_1 \cup \mathscr{S}_2$ with high probability.
\end{proof}

\begin{customlem}{\ref{Lemma3}}
Let $(X_1, X_2)$ follow an unidentifiable $CPCM(F_1, \dots, F_k)$ with DAG $\mathcal{G}$. Then, for sufficiently large $n$, Algorithm~\ref{Algorithm1} outputs “Unidentifiable case” with probability at least $1 - 2\alpha$, given that we employ a HSIC independence test in step 1b) of Algorithm~\ref{Algorithm1}, and assume access to an oracle regression estimator $\hat\theta = \theta$. 
\end{customlem}

\begin{proof}
If the true model is unidentifiable, independence holds in both directions: $X_1 \indep \varepsilon_2$ and $X_2 \indep \varepsilon_1$. Given an oracle regression estimator $\hat\theta = \theta$, we have $\hat\varepsilon_i = \varepsilon_i$ for $i = 1, 2$. Therefore, Algorithm 1 reduces to performing two HSIC independence tests and returning “Unidentifiable case” if both tests fail to reject the null.

The HSIC test has asymptotically correct level under the null hypothesis of independence (\cite{Kernel_based_tests}, Theorem 3.8). Thus, for sufficiently large sample size $n$, the probability that the HSIC test fails to reject the null in each direction is at least $1 - \alpha$. By the union bound, the probability that both tests fail to reject is at least $(1 - \alpha)^2 > 1 - 2\alpha$. Hence, Algorithm 1 returns “Unidentifiable case” with probability larger than $1 - 2\alpha$.
\end{proof}

\renewcommand\thesection{B}
\section{Experiments details and additional plots}
\label{Appendix_simulations}

\subsection{Exact, Naive-greedy, RESIT, and RESIT-greedy algorithms: definitions and comparison}
\label{appendix_greedy_definitions}
We consider the following algorithms for estimating the underlying causal graph from observational data using the CPCM score function~\eqref{score_definition1}:

\begin{itemize}
    \item \textbf{Exact search:} This algorithm evaluates the CPCM score for all DAGs on $d$ nodes and selects the one with the lowest score. Since the number of DAGs grows super-exponentially with $d$ (e.g., 29,281 DAGs for $d = 5$), exact search is computationally feasible only for graphs with $d \leq 4$.
    
    \item \textbf{Naive-edge-greedy:} Starting from an empty DAG, this algorithm iteratively explores neighboring DAGs by adding or removing a single edge. At each step, it selects the neighboring graph with the lowest CPCM score and replaces the current graph if the score improves. The procedure stops when no further improvement is possible. While simple and scalable, this greedy approach lacks theoretical guarantees and may get stuck in local minima, unless we assume some advanced notions of convexity over the space of all DAGs.

    \item \textbf{RESIT (Regression with Subsequent Independence Test):} RESIT first estimates a topological ordering by iteratively selecting the variable whose residual is least dependent on the remaining variables. In the second phase, it removes superfluous edges using conditional independence tests. See Algorithm~\ref{alg:resit-cpcm} for details. The procedure is computationally efficient and comes with statistical guarantees (see Lemma~\ref{thm:resit_consistency}). However, empirical performance in smaller sample sizes tends to be worse than that of greedy algorithms, particularly due to the accumulation of errors in the ordering phase and false positives (type I errors) in the Phase~2. 

    \item \textbf{RESIT-greedy:} This hybrid algorithm combines the topological ordering phase of RESIT with the edge-pruning phase of naive-greedy search. After estimating the ordering, it starts with a fully connected DAG consistent with the order and iteratively removes edges that lead to the largest improvement in the CPCM score, until no further improvement is possible. See Algorithm~\ref{Algorithm_RESIT_GREEDY}.
\end{itemize}

\citet{Peters2014} showed that, in the population case, the RESIT algorithm is consistent under identifiable additive noise models, assuming a consistent nonparametric regression method and a perfect independence oracle. The same reasoning applies directly to CPCM models.

\begin{lemma}[Consistency of RESIT under CPCM]
\label{thm:resit_consistency}
Let $\mathbf{X}$ be generated by an identifiable $CPCM(F_1, \dots, F_k)$ model with underlying DAG $\mathcal{G}_0$. Then, the RESIT algorithm, when applied with consistent estimators $\hat{\theta}_k$ and an independence oracle, is guaranteed to recover the true graph $\mathcal{G}_0$ from the distribution of $\mathbf{X}$.
\end{lemma}

\begin{proof}
A direct consequence of Theorem~34 in \citet{Peters2014}. Note that we implicitly assume the causal minimality condition for $CPCM(F_1, \dots, F_k)$ as stated in Definition~\ref{DefinitionCPCM}.
\end{proof}

\begin{algorithm}[H]
\caption{Regression with Subsequent Independence Test (RESIT; \cite{Peters2014}), modified for $CPCM(F_1, \dots, F_k)$}
\label{alg:resit-cpcm}
\KwIn{Random sample of $(X_1,\ldots,X_d)$; candidate families $F_1,\dots,F_k$}

\textbf{Pre-step (support gating).} 

For each node $i$, estimate
$\hat S_i := \{\,j  \in \{1,\dots,k\}:\ \mathrm{supp}(X_i)=\mathrm{supp}(F_j)\,\}$\tcp*{See Section~\ref{section5.1.1}. This filters discrete vs. continuous etc. }
If $\hat{S}_i = \emptyset$ return ``Assumptions not fulfilled''.

\medskip
\textbf{Phase 1: Determine a topological order $\pi$}

$M \gets \{1,\ldots,d\}$; $\pi \gets [\ ]$\;
\While{$M\neq\emptyset$}{
  \ForEach{$v \in M$}{
    \ForEach{$m \in \hat S_v$}{
      Fit parameters $\hat\theta_v^{(m)}$ in
      $X_v \mid X_{M\setminus\{v\}} \sim F_m(\,\theta_v(\cdot)\,)$\;
      Compute residuals
      $\hat\varepsilon_v^{(m)} := F_m\!\left(X_v;\ \hat\theta_v^{(m)}\!\big(X_{M\setminus\{v\}}\big)\right)$\;
      Compute dependence score 
      $s_v^{(m)} := \rho\!\left(\hat\varepsilon_v^{(m)}\,;\ X_{M\setminus\{v\}}\right)$
      \tcp*{$\rho$ small $\Leftrightarrow$ near independence. Implemention uses HSIC score}
    }
    $s_v \gets \min_{m\in\hat S_v} s_v^{(m)}$ \tcp*{Best fitting distribution $F_m$} 
  }
  $v^\star \gets \arg\min_{v\in M} s_v$ \tcp*{most source-like variable}
  $M \gets M\setminus\{v^\star\}$;\quad $\mathrm{pa}(v^\star) \gets M$\;
  prepend $v^\star$ to $\pi$
}

\medskip
\textbf{Phase 2: Prune superfluous edges}

\For{$t=2$ \KwTo $d$}{
  $v \gets \pi_t$;\quad $C \gets \mathrm{pa}(v)$\;
  \ForEach{$\ell \in C$}{
    \ForEach{$m \in \hat S_v$}{
      Fit $\hat\theta_{v,-\ell}^{(m)}$ in 
      $X_v \mid X_{C\setminus\{\ell\}} \sim F_m(\,\theta_v(\cdot)\,)$\;
      Compute residuals
      $\hat\varepsilon_{v}^{(m)} := 
      F_m\!\left(X_v;\ \hat\theta_{v,-\ell}^{(m)}\!\big(X_{C\setminus\{\ell\}}\big)\right)$\;
      Compute $p_\ell^{(m)} := \textsf{IndepP}\!\left(\hat\varepsilon_{v}^{(m)} \,;\ 
      \{X_{\pi_1},\ldots,X_{\pi_{t-1}}\}\right)$ \tcp*{p-value of the test of $H_0: \varepsilon_{v}^{(m)} \indep \{X_{\pi_1},\ldots,X_{\pi_{t-1}}\}$}
    }
    $p_\ell := \max_{m\in\hat S_v} p_\ell^{(m)}$\;
    \If{$p_\ell \ge \alpha$\text{ using }$\alpha=0.05$}{
      $C \gets C\setminus\{\ell\}$ \tcp*{remove $\ell$ if residuals are independent}
    }
  }
  $\mathrm{pa}(v) \gets C$
}
\KwOut{$(\mathrm{pa}(1),\ldots,\mathrm{pa}(d))$}
\end{algorithm}

\newpage
\subsubsection{Experiments: comparison of different greedy algorithms}
\label{appendix_greedy}
\textbf{Experimental setup:} We generate  random DAGs uniformly over $d$ nodes with $p$ edges, where $p \sim \mathrm{Exp}(1/d)$ and capped at $\frac{d(d-1)}{2}$. On average, each DAG contains approximately $d$ edges. For a given DAG $\mathcal{G}$, we simulate data from the $CPCM(F)$ model, $F=Exponential$, defined as follows:
$$
X_i \sim \mathrm{Exp}\left(\lambda(\mathbf{X}_{pa_i(\mathcal{G})})\right), \quad \text{with} \quad \lambda(\mathbf{X}_{pa_i(\mathcal{G})}) = \frac{1}{\sum_{j \in pa_i(\mathcal{G})} |X_j|} = \frac{1}{\mathbb{E}[X_i\mid \textbf{X}_{pa_i}]}.
$$
If $pa_i(\mathcal{G})=\emptyset$, then $X_i\sim N(0,1)$. Using $n = 1000$ samples, we estimate $\mathcal{G}$ using the score-based estimator with score defined in Equation~\eqref{multi_extention}, with a fixed function class $\mathscr{S}_1$ (rather than the sequential approach) to allow for a fair comparison in both accuracy and computational cost. We compare the Exact, Naive-greedy, RESIT, and RESIT-greedy methods.  We evaluate their performance using the Structural Intervention Distance (SID, \cite{peters2014structuralinterventiondistancesid}), computing $\mathrm{SID}(\mathcal{G}, \hat{\mathcal{G}})$ for each estimated graph.

\textbf{Results:} Figure~\ref{figure_greedy} presents the average normalized $\frac{SID}{d}$ over 50 repetitions.

\begin{itemize}
    \item The exact method achieves, unsurprisingly, the lowest SID, with the highest computational cost. 
    \item Both RESIT and naive-greedy exhibit similar performance. The greedy approach achieves slightly lower SID on average but requires slightly more computation time. Note that for $d\geq 8$, both methods perform as badly as Trivial algorithm (empty graph).
    \item RESIT-greedy serves as a middle ground: it significantly improves over RESIT/naive-greedy methods in terms of SID while incurring higher computational cost. 
\end{itemize}

These experiments highlight the trade-offs between statistical performance and computational efficiency among the evaluated methods. While the exact method yields the most accurate graph recovery, its scalability is limited. In contrast, RESIT and naive-greedy offer faster but less accurate alternatives, with performance deteriorating as graph complexity increases. RESIT-greedy provides a promising compromise, achieving lower SID than standard greedy methods at a moderate computational cost. \textbf{Overall, RESIT-greedy seems to be a practical choice in graphs with size} $4<d<10$.

\begin{figure}[]
\centering
\includegraphics[scale=0.62]{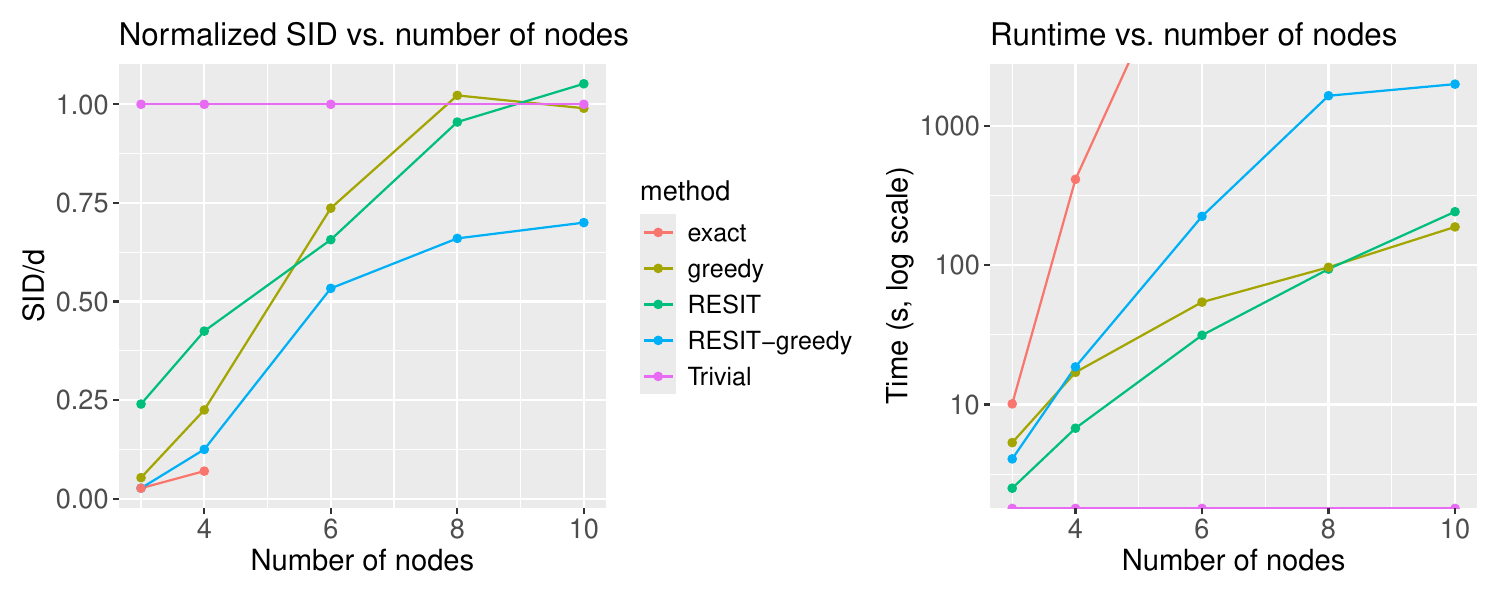}
\caption{Performance of different greedy algorithms from Section~\ref{appendix_greedy}, measured by the normalized Structural Interventional Distance (SID). Here, the \texttt{Trivial} algorithm always returns an empty graph. Runtime was measured on a machine with an Intel Core i5-6300U 2.5 GHz processor and 16 GB of RAM.}
\label{figure_greedy}
\end{figure}

\subsubsection{Statistical scalability: sample size needs to grow with increasing dimension}
\label{appendix_scalability}
Conducting independence tests and performing nonparametric regression becomes increasingly difficult as the number of covariates grows. In high-dimensional settings, these tasks require substantially larger sample sizes for reliable estimation. Similar limitations are observed in other methods such as ANM-RESIT, LOCI, and bQCD. The experiment below illustrates how the sample size required for consistent estimation increases systematically with the dimension~$d$.

\textbf{Experimental setup.} For each sample size \( n \) and dimension \( d \), we generated a uniformly random DAG with \( d \) nodes and \( d-1 \) edges using the \texttt{bnlearn} package in R. Data were drawn from the structural equation model \( X_i = f_i(\mathbf{X}_{pa_i}) + \varepsilon_i \), where \( \varepsilon_i \sim N(0,1) \) and \( f_i(x) = c_i x^\top x \) with \( c_i \sim \text{Unif}(0.5, 1.5) \). The DAG was estimated using the $CPCM(F)$ algorithm \eqref{score_definition1} with $F$ set to a Gaussian distribution of fixed variance, and structure learning was performed via the naive-greedy algorithm. Estimation accuracy was evaluated using the Structural Intervention Distance (SID). We also generate a baseline guess by generating random DAGs with \(d-1\) edges uniformly over the space of all DAGs. Each configuration was repeated 100 times.

\textbf{Results.} Figure~\ref{scalability} reports the ratio between the average SID obtained by the $CPCM(F)$ algorithm and that of a random-DAG baseline. As the dimension \(d\) increases, considerably larger samples are required to achieve ratios below \(0.5\), which correspond to meaningful structural recovery. Specifically, we observe ratios below \(0.5\) for approximately \(n=100\) when \(d=2\)–\(3\), \(n=500\) for \(d=4\), \(n=1000\) for \(d=5\), and \(n=2000\) for \(d=6\). These results indicate that \textbf{the sample size required for reliable estimation grows} \textbf{rapidly} (potentially exponentially) \textbf{with the dimension} \textbf{of} $\mathbf{X}$.

\begin{figure}[]
\centering
\includegraphics[scale=0.5]{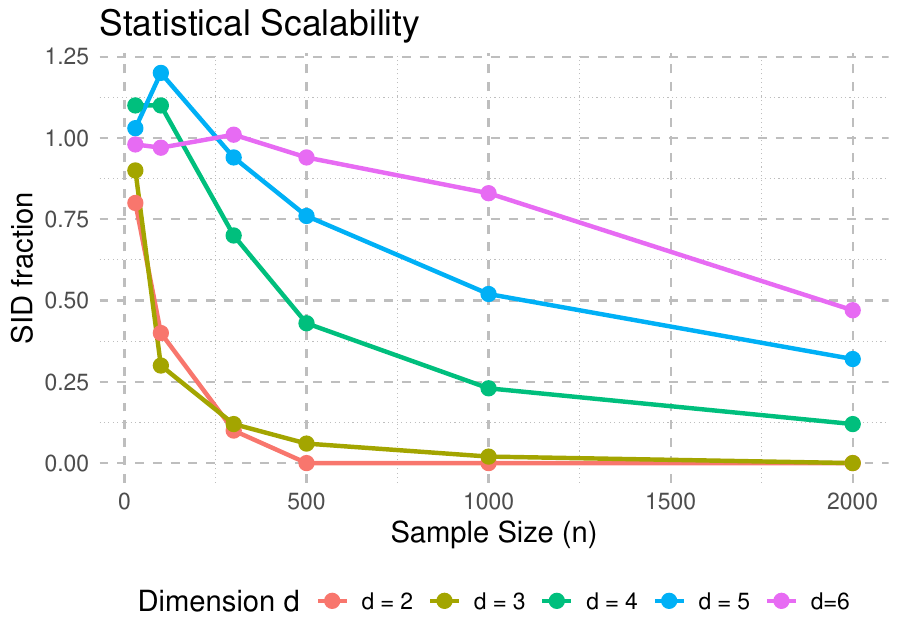}
\caption{
Scalability experiment from Section~\ref{appendix_scalability}. 
The plot shows the ratio between the average SID obtained by the $CPCM(F)$ algorithm (edge-greedy version) and that of a random-DAG baseline, across 100 repetitions for different values of \(n\) and \(d\). 
Lower ratios indicate better structural recovery, while values near one correspond to random performance. 
The results suggest that the sample size required for accurate estimation grows rapidly with the dimension of $\mathbf{X}$.
}
\label{scalability}
\end{figure}

\subsection{Sequential approach vs oracle: empirical performance}
\label{Simulations_seq_approach}
We evaluate the empirical performance of the sequential approach for selecting the distribution family, comparing it to $CPCM(F)$ with access to the true (oracle) distribution $F$. \\
\textbf{Experimental setup:} The data is generated as follows:
$$
X_1 \sim \mathcal{N}(0,1), \quad X_2 = F^{-1}(\varepsilon, \theta(X)), \quad \varepsilon \indep X, \; \varepsilon \sim \mathcal{U}(0,1),
$$
where $F \in \mathscr{S}_s$ belongs to either the one-parameter family $\mathscr{S}_1$ or the two-parameter family~$\mathscr{S}_2$. 

When $s = 1$, $F$ is, with equal probability ($\frac{1}{4}$), either a Gaussian distribution with fixed variance, a Poisson, Pareto, or Exponential distribution. When $s = 2$, $F$ is, again with equal probability, a Gaussian, Negative Binomial, Generalized Pareto, or Gamma distribution. All parameters are generated as a transformation of a random smooth polynomial using
\begin{equation*}
\theta(x)
= 0.5 + 4.5 \cdot
\frac{\tanh\big(b(x)^\top u\big) + 1}{2},
\qquad
u_j \sim \text{Unif}(-2,2),
\quad j=1,2,3,
\label{eq2345}
\end{equation*}
where \(b(x)\) denotes the natural spline basis with three degrees of freedom. This construction yields a smooth parameter function \(\theta(x) \in [0.5, 5]\), ensuring that the conditional distribution $X_2 \mid X_1$ varies continuously with~\(x\) and that the variance does not explode.

We estimate the graph $\mathcal{G} = {1 \to 2}$ using both $CPCM(\text{Seq.app})$ and $CPCM(F)$ with oracle knowledge of $F$, for various sample sizes $n$, repeating each experiment 100 times for both $s = 1$ and $s = 2$. Figure~\ref{Fig_seq} summarizes the results across sample sizes.

\textbf{Results.} For $s = 1$, we observe that $CPCM(\text{Seq.app})$ typically performs equivalently to oracle $CPCM(F)$ for $n > 100$. For $s = 2$, the sequential approach tends to select the simpler class $\mathscr{S}_1$ instead of the true two-parameter family $\mathscr{S}_2$ at smaller sample sizes. Nevertheless, the performance gap between the sequential approach and the oracle method remains almost negligible. These results indicate that \textbf{the sequential approach is performing nearly as well as the oracle} $CPCM(F)$, provided that $F \in \mathscr{S}_1 \cup \mathscr{S}_2$.

\begin{figure}[]
\centering
\includegraphics[scale=0.7]{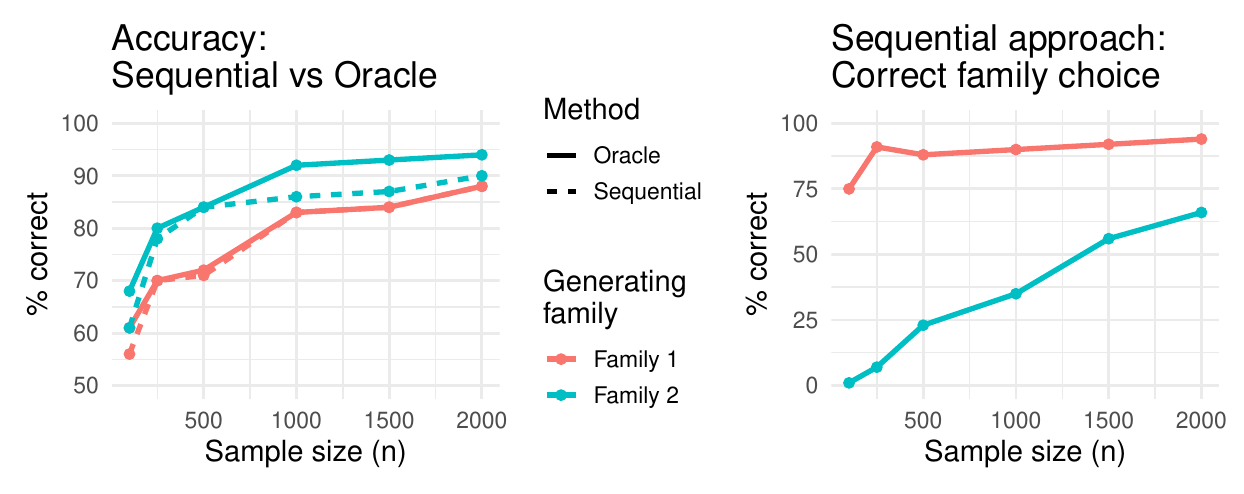}
\caption{Performance of the sequential approach. Left: percentage of simulations where $\hat{\mathcal{G}} = {1 \rightarrow 2}$. Right: percentage of simulations in which $CPCM(\text{Seq.app})$ and $CPCM(F)$ with oracle $F$ are equivalent. }
\label{Fig_seq}
\end{figure}

\subsection{Details about sections~\ref{Section_simulations_robustness}, \ref{Section_simulations_Pareto}, \ref{Section_simulations_Gaussian}, \ref{Section_simulations_multivariate} and \ref{Section7}}
\label{Appendix_Section_simulations_Pareto}

In \textbf{Section~\ref{Section_simulations_robustness}}, we generated random functions \(\theta(x)\) from zero-mean Gaussian processes with squared-exponential covariance kernel
\(
k(x,x') = \sigma^2 \exp\left(-\frac{(x-x')^2}{2\ell^2}\right),
\)
where \(\sigma = 1\) and \(\ell = 0.2\); same choices as in Section~\ref{Section_simulations_Gaussian} taken from \cite{Natasa_Tagasovska}. Each realization was then shifted to ensure positivity, \(\theta(x) \leftarrow \theta(x) - \min_x \theta(x) + 0.5\).

For \textbf{Section~\ref{Section_simulations_Pareto}}, the additional plots are presented in Table~\ref{Pareto_simulations1} and Figures~\ref{Pareto_histograms}~and~\ref{sample_size_pareto}. 

% \usepackage{multirow}
% \usepackage{graphicx}
\begin{table}[]
\centering
\renewcommand{\arraystretch}{1.15}
\begin{tabular}{l|
                S[table-format=3.0]
                S[table-format=3.0]
                S[table-format=3.0]
                S[table-format=3.0]
                S[table-format=3.0]}
\toprule
\textbf{$\gamma$} &
{$X_1 \to X_2$} &
{$X_2 \to X_1$} &
{\makecell{Empty\\graph}} &
{\makecell{Both directions\\appear plausible}} &
{\makecell{Neither direction\\appears plausible}} \\
\midrule
\rowcolor{RowAlt}
$-2$ & 0  & 0  & \bfseries 96 & 2  & 2  \\
$-1$ & 3  & 2  & 0  & \bfseries 95 & 0  \\
\rowcolor{RowAlt}
\,\,$0$  & 7  & 1  & 0  & \bfseries 92 & 0  \\
\,\,$1$  & 7  & 5  & 0  & \bfseries 86 & 2  \\
\rowcolor{RowAlt}
\,\,$2$  & \bfseries 93 & 0  & 0  & 14 & 3  \\
\bottomrule
\end{tabular}
\caption{Simulation results for the CPCM model using the Pareto distribution function \(F\). The table displays the percentage of cases for each type of graph structure estimated by Conservative Algorithm~\ref{Algorithm1} with the model specified in \eqref{fwesef}, across various values of the hyperparameter \(\gamma \in \mathbb{R}\). The columns indicate the frequency of each graph structure being estimated, with the highest frequency in each row highlighted in bold.}
\label{Pareto_simulations1}
\end{table}

\begin{figure}[]
\centering
\includegraphics[width = 0.8\textwidth]{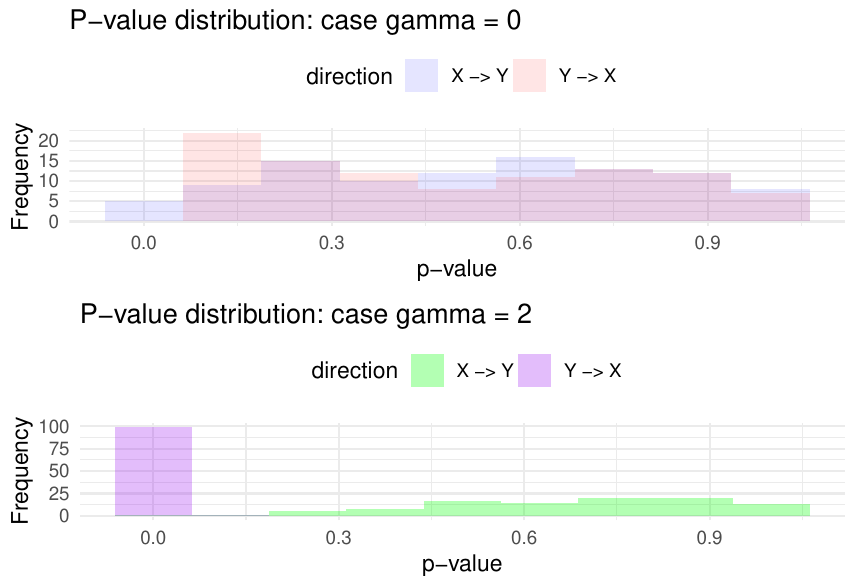}
\caption{(Simulations~\ref{Section_simulations_Pareto}). Distributions of the p-values from the independence test in Step 1b) of Algorithm~\ref{Algorithm1}, for model \eqref{fwesef} with \(\gamma = 0\) and \(\gamma = 2\).}
\label{Pareto_histograms}
\end{figure}

\begin{figure}[]
\centering
\includegraphics[scale=0.7]{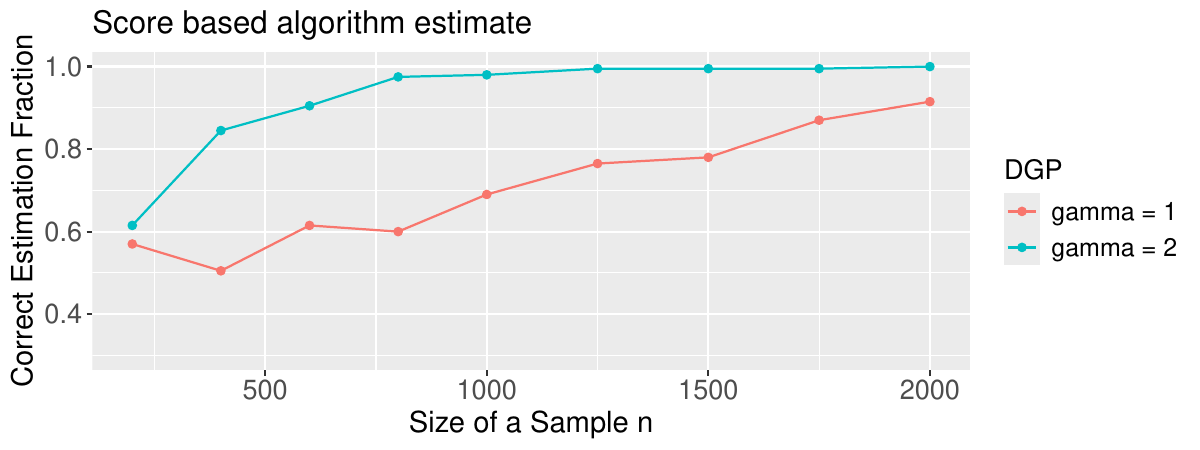}
\caption{(Simulations~\ref{Section_simulations_Pareto}). The plot displays the percentage of correctly estimated causal directions across a range of sample sizes \( n \), using model \eqref{fwesef} with hyperparameters \(\gamma = 1\) and \(\gamma = 2\). As \( n \) increases, the algorithm demonstrates near-perfect performance, affirming the theoretical consistency of the proposed method. }
\label{sample_size_pareto}
\end{figure}

In \textbf{Section~\ref{Section_simulations_Gaussian}}, the experiments were inspired by \cite{Natasa_Tagasovska} and implementations of other baseline methods are also taken from \cite{Natasa_Tagasovska} and \cite{immer2022identifiability}. 

\begin{itemize}
    \item For LOCI, we use the default format with neural network estimations and subsequent independence testing (also denoted as $NN-LOCI_H$) \citep{immer2022identifiability}.
    \item For IGCI, we use the original implementation from \cite{IGCI} with slope-based estimation with Gaussian and uniform reference measures. 
    \item For RESIT, we use the GP regression and the HSIC independence test with a threshold value of $\alpha = 0.05$.
    \item For the slope algorithm, we use the implementation of \cite{Slope}, with the local regression included in the fitting process.
    \item For comparisons with other methods such as PNL, GPI-MML, ANM, Sloppy, GR-AN, EMD, GRCI, see Section 3.2 in \cite{Natasa_Tagasovska} and Section 5 in \cite{immer2022identifiability}. 
\end{itemize}
The random functions were generated in the same way as in \cite{Natasa_Tagasovska}. Specifically, Models 1, 3, 6, and 7 are realizations of zero-mean Gaussian processes with a squared-exponential covariance kernel 
\(
k(x,x') = \sigma^2 \exp\left(-\frac{(x-x')^2}{2\ell^2}\right),
\)
where \(\sigma = 1\) and \(\ell = 0.2\). Each realization was then shifted to ensure positivity, \(\theta(x) \leftarrow \theta(x) - \min_x \theta(x) + 0.5\). 
Models 2, 4, and 5 use an injective (monotone) nonlinear transformation $m(x) = C\,\frac{B(x + A)}{1 + |B(x + A)|},$ where \(A \sim \mathrm{Unif}[-2,2]\), \(B\) follows a two–point mixture with \(B \sim \mathrm{Unif}[0.5,2]\) with probability $0.5$ and \(B \sim \mathrm{Unif}[-2,-0.5]\) otherwise, and \(C \sim 1 + \mathrm{Exp}(4)\).
This parametrization produces smooth saturating nonlinear relationships that can be either increasing or decreasing, thereby capturing a wide range of monotone functional dependencies between the parent and child variables. Figure~\ref{Simulations2_plots} shows an example of datasets generated via different models from Section~\ref{Section_simulations_Gaussian}.

In \textbf{Section~\ref{Section_simulations_multivariate}}, all baseline methods are implemented using the \texttt{pcalg} package \citep{pcalg_package}: 
\begin{itemize}
\item For the PC algorithm, we consider three variants: (1) the default Gaussian conditional independence test \texttt{gaussCItest}, (2) the HSIC-based test \citep{Kernel_based_tests}, and (3) the copula-based test \citep{copula_based_independence_test}, which is omitted from the table as it yielded consistently inferior results. We always used significance level $\alpha = 0.05$. 
\item The GES algorithm is implemented with the Gaussian observational BIC score and the default penalty parameter. 
\item LiNGAM with default parameters.
\item For the ANM baseline, we ensure fairness by adopting the same components as in our CPCM implementation; namely, the GAM estimator and the HSIC dependence measure.
\end{itemize}

Regarding \textbf{Section~\ref{Section7}},  Table~\ref{tab:edge_share} shows the relative frequencies with which each edge was recovered across repeated subsamples of the motor insurance dataset.

\begin{figure}[]
\centering
\includegraphics[scale=0.4]{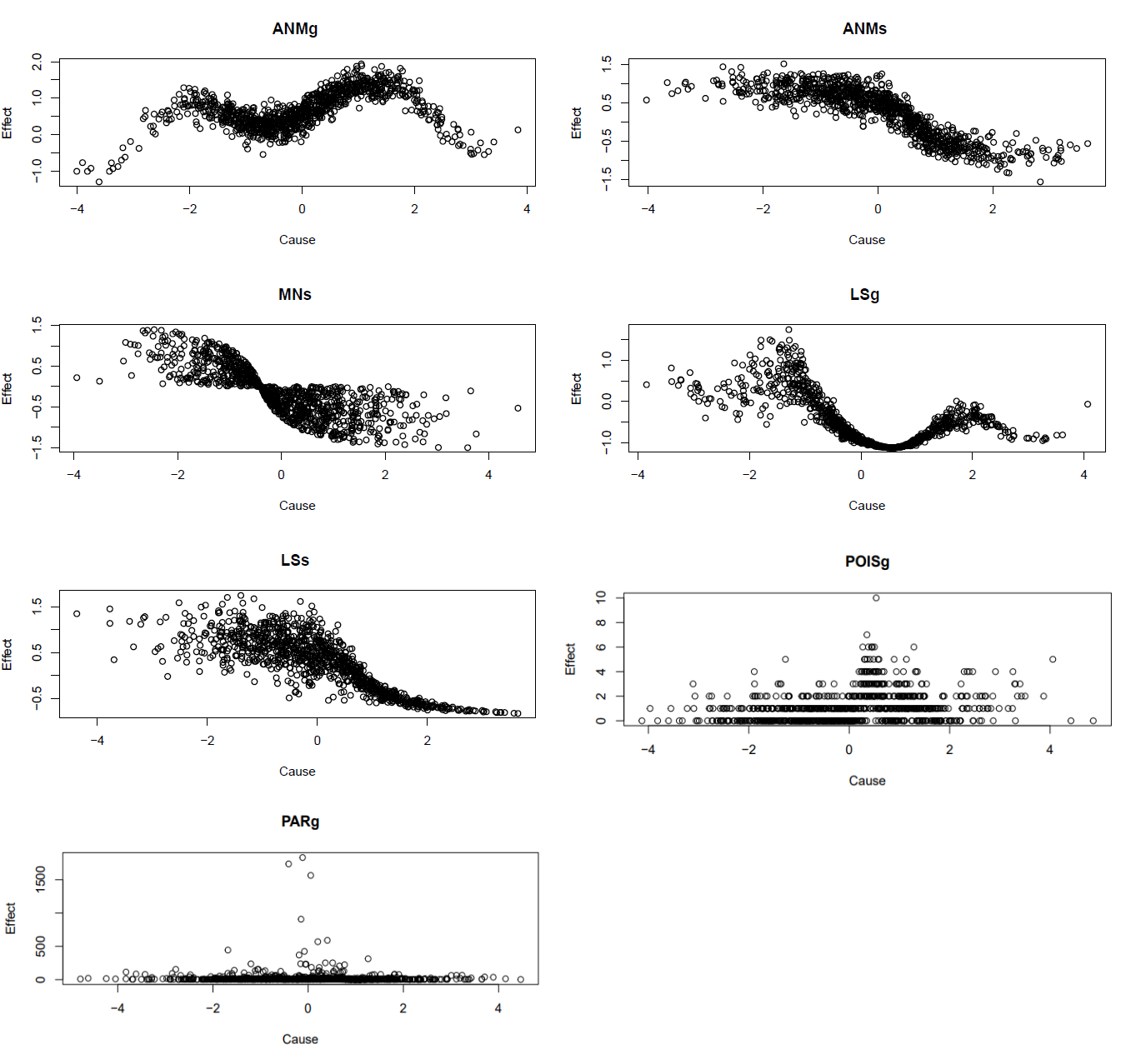}
\caption{Simulations~\ref{Section_simulations_Gaussian}. An example of datasets generated via different models. }
\label{Simulations2_plots}
\end{figure}

\begin{table}[ht]
\centering
\begin{minipage}{0.45\linewidth}
\centering
\begin{tabular}{lc}
\hline
Edge & Share (\%) \\
\hline
VehAge $\to$ ClaimNb     & 60 \\
Exposure $\to$ ClaimNb   & 48 \\
VehPower $\to$ ClaimNb   & 42 \\
VehPower $\to$ Exposure  & 44 \\
\hline
\end{tabular}
\end{minipage}\hfill
\begin{minipage}{0.45\linewidth}
\centering
\begin{tabular}{lc}
\hline
Edge & Share (\%) \\
\hline
VehAge $\to$ Exposure    & 36 \\
Exposure $\to$ VehAge    & 64 \\
VehPower $\to$ VehAge    & 54 \\
VehAge $\to$ VehPower    & 30 \\
\hline
\end{tabular}
\end{minipage}
\caption{Relative frequency (in \%) with which each directed edge was recovered by CPCM across 50 random subsamples of the French MTPL motor insurance dataset (shown only those with more than $25\%$ share).}
\label{tab:edge_share}
\end{table}

\renewcommand\thesection{C}
\section{Proofs}
\label{SectionProofs}

\subsection{Proof of Theorem~\ref{normalidentifiability}}
\begin{customthm}{\ref{normalidentifiability}}
Let $(X_1,X_2)$ admit the $CPCM(F)$ model with graph $X_1\to X_2$, where $F$ is the Gaussian distribution function with parameters $\theta(X_1)=\big(\mu(X_1), \sigma(X_1)\big)^\top$. 
Let $p_{\varepsilon_1}$ be the density of $\varepsilon_1$ that is absolutely continuous with full support $\mathbb{R}$. Let $\mu(x), \sigma(x)$ be two times differentiable.  

Then, the causal graph is identifiable from the joint distribution if and only if  there do not exist $a,c ,d,e, \alpha, \beta\in\mathbb{R}$,  
$a\geq 0,c>0, \beta>0$, such that
\begin{equation}\tag{\ref{norm}}
\frac{1}{\sigma^2(x)}=ax^2 + c, \,\,\,\,\,\,\,\,\,\,\,\,\,\,\,\,\,\,\,\,\, \frac{\mu(x)}{\sigma^2(x)}=d+ex,
\end{equation}
for all $x\in\mathbb{R}$ and
\begin{equation}\tag{\ref{DensityDEF}}
p_{\varepsilon_1}(x) \propto \sigma(x)e^{-\frac{1}{2}\big[ \frac{(x-\alpha)^2}{\beta^2}  - \frac{\mu^2(x)}{\sigma^2(x)}\big]},
\end{equation}
where $\propto$ represents an equality up to a constant (here, $p_{\varepsilon_1} $  is a valid density function if and only if $\frac{1}{\beta^2}\neq  \frac{e^2}{c}\mathbbm{1}[a=0]$). 
Specifically, if $\sigma(x)$ is constant (case $a=0$), then the causal graph is identifiable unless $\mu(x)$ is linear and $p_{\varepsilon_1}$ is the Gaussian density.
\end{customthm}

\begin{proof}
\label{Proof of normalidentifiability}{}
We opt for proving this theorem from scratch, without using Theorem \ref{thmAssymetricMultivariatesufficient}. An interested reader can try to use Theorem \ref{thmAssymetricMultivariatesufficient} instead. For clarity regarding the indexes, we use the notation $X=X_1, Y=X_2$. 

First, we show that if the causal graph is not identifiable, then $\mu(x)$ and $\sigma(x)$ must satisfy (\ref{norm}). Let $p_{(X,Y)}$ be the density function of $(X,Y)$. Since the causal graph is not identifiable, there exist two CPCM models that generate $p_{(X,Y)}$: the CPCM model with  $X\to Y$  and the function  $\theta(x)=\big(\mu(x), \sigma^2(x)\big)^\top$ and the CPCM model with  $Y\to X$ and the function $\tilde{\theta}(y)=\big(\tilde{\mu}(y), \tilde{\sigma}^2(y)\big)^\top$. 

We decompose (corresponding to the direction $X\to Y$) 
\begin{equation*}
p_{(X,Y)}(x,y) = p_{X}(x)p_{Y\mid X}(y\mid x) = p_{X}(x) \phi\big(y;\theta(x)\big),
\end{equation*}
where $\phi\big(y;\theta(x)\big)$ is the Gaussian density function with parameters $\theta(x) = \big(\mu(x), \sigma^2(x)\big)^\top$. We rewrite this in the other direction: 
$$
p_{(X,Y)}(x,y) = p_{Y}(y)p_{X\mid Y}(x\mid y) = p_{Y}(y) \phi\big(x;\tilde{\theta}(y)\big).
$$
We take the logarithm of both equations and rewrite them in the following manner:
\begin{equation*}
\log[p_{X}(x)] +  \log\bigg\{\frac{1}{\sqrt{2\pi \sigma^2(x)}}e^{\frac{-[y-\mu(x)]^2}{2\sigma^2(x)}}\bigg\} = \log[p_{Y}(y)] +  \log\bigg\{\frac{1}{\sqrt{2\pi \tilde{\sigma}^2(y)}}e^{\frac{-(x-\tilde{\mu}(y))^2}{2\tilde{\sigma}^2(y)}}\bigg\} \text{  and}
\end{equation*}
\begin{equation}\label{eq1}
\log[p_{X}(x)] -\log\sigma(x)-\frac{1}{2}  \frac{[y-\mu(x)]^2}{\sigma^2(x)} = \log[p_{Y}(y)] -\log\tilde{\sigma}(y) -\frac{1}{2}  \frac{[x-\tilde{\mu}(y)]^2}{\tilde{\sigma}^2(y)}. 
\end{equation}
Calculating on both sides $\frac{\partial^4 }{\partial^2 x \partial^2 y }$, we obtain 
$$\frac{\sigma''(x)\sigma(x)-3\sigma'(x)'\sigma(x)}{\sigma^4(x)} =
\frac{\tilde{\sigma}''(y)\tilde{\sigma}(y)-3\tilde{\sigma}'(y)\tilde{\sigma}'(y)}{\tilde{\sigma}^4(y)}. $$Since this has to hold for all $x,y$, both sides need to be constant (let us denote this constant by $a\in\mathbb{R}$).  

Differential equation $\sigma''(x)\sigma(x)-3\sigma'(x)\sigma'(x)=a\,\sigma^4(x)$ has solution $\sigma(x) = \frac{1}{\sqrt{a(x+b)^2 + c}}$ for  $x$ , such that $a(x+b)^2 + c>0$. 

Plugging this result into (\ref{eq1}) and calculating on both sides $\frac{\partial^3 }{\partial^2 x \partial y }$, we obtain 
\begin{equation}\label{eq2}
\mu''(x) (a(x+b)^2+c) + \mu'(x) (4ax+4ab) + \mu(x) 2a = 2ab.
\end{equation}
Equation (\ref{eq2}) is another differential equation with a solution $\mu(x) = \frac{d+ex}{a(x+b)^2+c} + b$, for some $d,e\in\mathbb{R}$ for all $x{:}\,\, \sigma(x)>0$. 

Next, we show that it is necessary that $b=0$. If we show $b=0$, then $\mu(x)$ and $\sigma^2(x)$ are exactly in the form (\ref{norm}). We plug the representations  $\mu(x) = \frac{d+ex}{a(x+b)^2+c} + b, \sigma(x) = \frac{1}{\sqrt{a(x+b)^2 + c}}$ and   $\tilde{\mu}(x) = \frac{\tilde{d}+\tilde{e}x}{\tilde{a}(x+\tilde{b})^2+\tilde{c}} + \tilde{b}, \tilde{\sigma}(x) = \frac{1}{\sqrt{\tilde{a}(x+\tilde{b})^2 + \tilde{c}}}$  into (\ref{eq1}). Thus, we obtain
\begin{equation*}
\begin{split}
&\log[p_{X}(x)] +\frac{1}{2}\log[a(x+b)^2 + c]\\&
-\frac{1}{2}  \bigg[y^2(ax^2 + 2abx + ab^2+c) + y(2d+2ex) + \frac{1}{a(x+b)^2 +c} \bigg] \\&
= \log[p_{Y}(y)] +\frac{1}{2}\log[\tilde{a}(y+\tilde{b})^2 + \tilde{c}]\\& 
-\frac{1}{2}  \bigg[x^2(\tilde{a}y^2 + 2\tilde{a}\tilde{b}y + \tilde{a}\tilde{b}^2+\tilde{c}) + x(2\tilde{d}+2\tilde{e}y) + \frac{1}{\tilde{a}(y+\tilde{b})^2 +\tilde{c}} \bigg] .
\end{split}
\end{equation*}
We can re-write the last expression as
\begin{equation}\label{eq4}
\begin{split}
&h_X(x) + h_Y(y) =\frac{1}{2}  [y^2(ax^2 + 2abx ) + 2yex ] -\frac{1}{2}  [x^2(\tilde{a}y^2 + 2\tilde{a}\tilde{b}y) + 2x\tilde{e}y) ]\\
&=\frac{1}{2}[ x^2y^2(a - \tilde{a})  + xy(2aby - 2\tilde{a}\tilde{b}x +e-\tilde{e}) ],
\end{split}
\end{equation}
where 
\begin{equation*}\label{eqh_x}
h_X(x) = \log[p_{X}(x)]  +  \frac{1}{2}\log[a(x+b)^2 + c] + \frac{1}{2}\frac{1}{a(x+b)^2 +c} -2x\tilde{d} - x^2(\tilde{a}\tilde{b}^2+\tilde{c}), 
\end{equation*}
\begin{equation*}
h_Y(y) = -\log[p_{Y}(y)] -\frac{1}{2}\log[\tilde{a}(y+\tilde{b})^2 + \tilde{c}] +\frac{1}{2}[\frac{1}{\tilde{a}(y+\tilde{b})^2 +\tilde{c}} - 2yd - y^2(ab^2+c)].
\end{equation*}
Since the left-hand side of (\ref{eq4}) is in additive form, the right side also needs to have an additive representation. However, that is only possible if $a- \tilde{a}=0$ and $2aby - 2\tilde{a}\tilde{b}x +e-\tilde{e}=0$. Therefore, we necessarily have $a=\tilde{a}$ and either $a=0$ or $b=\tilde{b}=0$. The case $a=0$ corresponds to a constant $\sigma$ and, hence, also $b=\tilde{b}=0$. We have shown that $\mu(x)$ and $\sigma^2(x)$ have to satisfy (\ref{norm}).

Next, we show that if the causal graph is not identifiable, then the density of $p_{X}(x)$ has form (\ref{DensityDEF}). Plugging the form of  $\mu(x)$ and $\sigma^2(x)$ into (\ref{eq1}), we obtain
\begin{equation*}
\begin{split}
\log[p_{X}(x)] -&\log\sigma(x)-\frac{1}{2}\bigg(y-  \frac{d+ex}{ax^2+c}    \bigg)^2(ax^2+c) \\&= \log[p_{Y}(y)] -\log\tilde{\sigma}(y) -\frac{1}{2}  \bigg(x-  \frac{\tilde{d}+\tilde{e}y}{ay^2+\tilde{c}}   \bigg)^2(ay^2 + \tilde{c}).
\end{split}
\end{equation*}
We rewrite
\begin{equation}\label{rftgyh}
\begin{split}
\log[p_{X}(x)] -  &\log\sigma(x)+\frac{1}{2}\bigg[ \tilde{c}x^2 - 2x\tilde{d} +  \frac{({d}+{e}x)^2}{ax^2+{c}}   \bigg] \\&= \log[p_{Y}(y)] -\log\tilde{\sigma}(y) +\frac{1}{2}\bigg[cy^2 - 2yd + \frac{(\tilde{d}+\tilde{e}y)^2}{ay^2+\tilde{c}} \bigg]
. 
\end{split}
\end{equation}
Since this has to hold for all  $x,y\in\mathbb{R}$, both sides of (\ref{rftgyh}) need to be constant and we obtain $\log[p_{X}(x))]\propto\log\sigma(x)-\frac{1}{2}  \big[\tilde{c}x^2 - 2x\tilde{d} +  \frac{({d}+{e}x)^2}{ax^2+{c}}   \big] $. Hence, 
$$
p_{X}(x) \propto \sigma(x)e^{-\frac{1}{2}  \big[ \tilde{c}x^2 - 2x\tilde{d} +  \frac{({d}+{e}x)^2}{ax^2+{c}}   \big]} =  \sigma(x)e^{-\frac{1}{2}\big[ \frac{(x-\alpha)^2}{\beta^2}  - \frac{\mu^2(x)}{\sigma^2(x)}\big]},
$$
where $\beta = 1/\sqrt{\tilde{c}}$ and $\alpha=\frac{\tilde{d}}{\tilde{c}}$. The condition $\frac{1}{\beta^2} > \frac{e^2}{c}\mathbbm{1}[a=0]$ arises from the fact that if $a=0$ and $\frac{1}{\beta^2} \leq \frac{e^2}{c}$, then $p_X(x)$ is not a density function. This is because  
\[ 
\sigma(x)e^{-\frac{1}{2}\left[ \frac{(x-\alpha)^2}{\beta^2} - \frac{\mu^2(x)}{\sigma^2(x)}\right]} \propto e^{-\frac{1}{2}\left[ x^2\left(\frac{1}{\beta^2} - \frac{e^2}{c}\right) + x\left(-2\frac{\alpha}{\beta^2} - \frac{2de}{c}\right) \right] }
\]
for all \(x \in \mathbb{R}\). This expression is integrable is and only if the coefficient at \(x^2\) is positive.

Finally, we deal with the other direction: we show that if $\mu$ and $\sigma$ satisfy (\ref{norm}) and $p_{\varepsilon_X}$ has form (\ref{DensityDEF}), then the causal graph is not identifiable. Assume that $a,c,d,e$ are given. Define $\tilde{a} = a, \tilde{e}=e$ and select $\tilde{c}, \tilde{d}\in\mathbb{R}$, such that $\tilde{c}>0,\tilde{c}\neq  \frac{e^2}{c}\mathbbm{1}[a=0]$. Define $\frac{1}{\tilde{\sigma}^2(y)} = \tilde{a} y^2 + \tilde{c}, \frac{\tilde{\mu}(y)}{\tilde{\sigma}^2(y)} = \tilde{d}+\tilde{e}y$. Moreover, define 
\begin{align*}
p_X(x) &\propto \sigma(x)e^{-\frac{1}{2}\big[ \tilde{c}\big(x-\frac{\tilde{d}}{\tilde{c}}\big)^2  - \frac{\mu^2(x)}{\sigma^2(x)}\big]}\,\,\text{             and}\\
p_Y(y) &\propto \tilde{\sigma}(y)e^{-\frac{1}{2}\big[c\big(x-\frac{{d}}{{c}}\big)^2  - \frac{\tilde{\mu}^2(y)}{\tilde{\sigma}^2(y)}\big]}.
\end{align*}
Note that regardless of the coefficients, these are valid density functions (with one exception when $\tilde{c}=\frac{e^2}{c}$ and $a=0$, which is why we selected $\tilde{c}\neq  \frac{e^2}{c}\mathbbm{1}[a=0]$). In case of $a=0$, this is the classical Gaussian distribution density function.

Using these values, we obtain the equality
$$
p_X(x)p_{Y\mid X}(y\mid x) = p_Y(y)p_{X\mid Y}(x\mid y), \forall x,y\in\mathbb{R},
$$
or more precisely, 
$$
\sigma(x)e^{-\frac{1}{2}\big[ \tilde{c}\big(x-\frac{\tilde{d}}{\tilde{c}}\big)^2  - \frac{\mu^2(x)}{\sigma^2(x)}\big]}\frac{1}{\sqrt{2\pi}\sigma(x)}e^{-\frac{1}{2}\frac{[y-\mu(x)]^2}{\sigma^2(x)}} \propto \tilde{\sigma}(y)e^{-\frac{1}{2}\big[c\big(x-\frac{{d}}{{c}}\big)^2  - \frac{\tilde{\mu}^2(y)}{\tilde{\sigma}^2(y)}\big]}\frac{1}{\sqrt{2\pi}\tilde{\sigma}(y)}e^{-\frac{1}{2}\frac{(x-\tilde{\mu}(y))^2}{\tilde{\sigma}^2(y)}}. 
$$
Since this holds for all $x,y\in\mathbb{R}$, we found a valid backward model. The density in (\ref{DensityDEF}) uses the notation $\alpha=\frac{\tilde{d}}{\tilde{c}}$ and $\beta = 1/\sqrt{\tilde{c}}$.  
\end{proof}

An example of the joint distribution of $X_1, X_2$ with $a=c=d=e=\alpha = \beta=1$ is depicted in Figure \ref{GaussianDensity}. 
\begin{figure}[ht]
\centering
\includegraphics[scale=0.35]{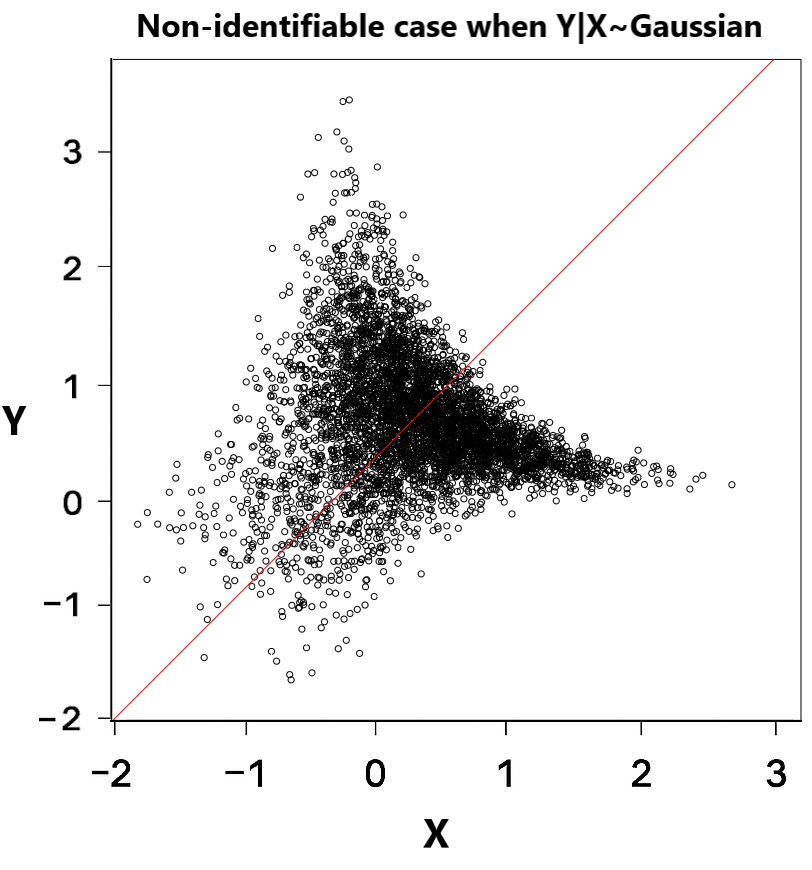}
\caption{Random sample from a joint distribution of $(X_1,X_2)$, where  $X_1$ has the marginal density (\ref{DensityDEF}) and $X_2\mid X_1\sim N\big(\mu(X_1), \sigma^2(X_1)\big)$ with $\mu, \sigma$ defined in (\ref{norm}) with constants   $a=c=d=e=\alpha= \beta=1$. The distribution function is symmetric according to the $x=y$ axis (red line).    }
\label{GaussianDensity}
\end{figure}

\subsection{Proof of Consequence \ref{paretoidentifiability}}

\begin{customconsequence}{\ref{paretoidentifiability}}
\begin{itemize}
   \item Let $(X_1,X_2)$admit the $CPCM(F)$ model with graph $X_1\to X_2$, where $F$ is the (discrete) Poisson distribution function.  Then, the causal graph is \textit{not} identifiable if and only if
\begin{equation}\label{eq505}
     \lambda(x) =e^{ax+b},\,\,\,\,\,\,\,\, P(X_1=x) \propto \frac{e^{e^{ax+b}+cx}}{x! }, \,\,\,\,\,\,\,\,\forall x\in\mathbb{N}_0,
     \end{equation}
for some $a<0,b,c\in\mathbb{R}$.

   \item  Let $(X_1,X_2)$ admit the $CPCM(F)$ model with graph $X_1\to X_2$, where $F$ is the Pareto distribution function. Then, the causal graph is \textit{not} identifiable if and only if   
\begin{equation}\label{eq50}
\theta(x) = a\log(x) +b,\,\,\,\,\,\,\,\, p_{X_1}(x) \propto \frac{1}{ [a\log(x)+b] x^{c+1} }, \,\,\,\,\,\,\,\,\forall x\geq 1,
\end{equation}
for some $a,b,c>0$.

    \item Let $(X_1,X_2)$ admit the $CPCM(F)$ model with graph $X_1\to X_2$, where $F$ is Bernoulli distribution function. Then, the causal graph is identifiable if and only if $supp(X_1) \neq \{0,1\}$.
\end{itemize}

\end{customconsequence}
\begin{proof}
\label{Proof of pareto identifiability}:
\textbf{First bullet-point: }Poisson distribution has one parameter and it can be written as $h_1(x) = 1/x!, h_2(x) = e^{-e^x}, T(x) = x$. Note that we do not use classical form of density function but its reparametrisation where $\theta(x) = \log(\lambda(x))$ where $\lambda$ is the classical rate parameter. 

Plugging this into Proposition~\ref{Necessary condition for identifiability}, we directly obtain (\ref{eq505}) with possible $a\in\mathbb{R}$. It is not hard to see that $ \frac{e^{e^{ax+b}+cx}}{x! }$ is integrable if and only if $a<0$. In such a case, the backward model exist and has a form $\tilde{\theta}(y) = a y+c$ and $P(X_2=y) \propto \frac{e^{e^{ay+c}+by}}{y! } $.

\textbf{Second bullet-point:} Pareto distribution has one parameter and it can be written as $h_1(x) = 1, h_2(\theta) = 1/\theta, T(x) = log(x)$. 

Plugging this into Proposition~\ref{Necessary condition for identifiability}, we directly obtain (\ref{eq50}) with possible $a,b,c\in\mathbb{R}$.  It is not hard to see that the density is integrable if and only if $a,c,b>0$, in which case, the backward model exist and has a form $\tilde{\theta}(y) = \tilde{a}\log(y) + \tilde{b}$, $ p_{X_2}(x) \propto \frac{1}{ [\tilde{a}\log(x)+\tilde{b}] x^{\tilde{c}+1} }$ for $\tilde{a}=a,\tilde{b} = c,\tilde{c} = b$. 

\textbf{Third bullet-point:} If $supp(X_1) \neq \{0,1\}$, then the causal graph is identifiable due to the first bullet-point in Proposition~\ref{Necessary condition for identifiability}. Next, consider $supp(X_1) = \{0,1\}$. In this case, we can always write a backward model for the Bernoulli distribution. Let $P(X_1=X_2=0) = p_0$,  $P(X_1=0, X_2 = 1) = p_{0,1}$, $P(X_1=1, X_2 = 0) = p_{1,0}$ and $P(X_1=X_2 = 1) = p_{1}$ for $p_0, p_{0,1}, p_{1,0}, p_1>0$ and $p_0 + p_{0,1} + p_{1,0} + p_1 = 1$. We can define $X_2\mid X_1\sim Bernoulli(\theta(X_1))$ as $\theta(0) = p_{0,1}$ and $\theta(1) = p_1$. On the other hand, we can define $X_1\mid X_2 \sim Bernoulli(\tilde{\theta}(X_2))$ as $\tilde{\theta}(0) = p_{1,0}$ and $\tilde{\theta}(1) = p_{1}$. Since both models produce the same joint distribution, the causal model is not identifiable for any values of $p_0, p_{0,1}, p_{1,0}, p_1$.  
\end{proof}

%%%%%%%%%%%%%%%%%%%%% Theorem 2 %%%%%%%%%%%%%%%%%%%%
\subsection{Proof of Theorem \ref{thmAssymetricMultivariatesufficient}}
Before we prove Theorem \ref{thmAssymetricMultivariatesufficient}, we show the following auxiliary lemma. 
\begin{lemma}\label{PomocnaLemma1}
Let $n\in\mathbb{N}$ and $\mathcal{X,Y}\subseteq \mathbb{R}$. Let $f_1, \dots, f_n, g_1, \dots, g_n$ be non-constant functions on $\mathcal{X,Y}$, respectively, such that 
$
f_1(x)g_1(y) + \dots + f_n(x)g_n(y)
$ is additive in $x,y$---that is, there exist functions $f$ and $g$, such that 
$$
f_1(x)g_1(y) + \dots + f_n(x)g_n(y) = f(x) + g(y), \forall x\in\mathcal{X},y\in\mathcal{Y}.
$$
Then, there exist (not all zero) constants $a_1, \dots, a_n, c\in\mathbb{R}$, such that 
$\sum_{i=1}^n a_if_i(x) = c$ for all $x\in\mathcal{X}$. Specifically for $n=2$, it holds that $f_1(x) = af_2(x)+c$ for some $a,c\in\mathbb{R}$. 

Moreover, assume that for some $q<n$, functions $g_1, \dots, g_q$ are affinly independent---that is, there exist $y_0, y_1, \dots, y_q\in\mathcal{Y}$, such that a matrix 
\begin{equation}\label{matrix243}
M=\begin{pmatrix}
 g_1(y_1) - g_1(y_0) & \cdots & g_q(y_1) - g_q(y_0) \\
\cdots & \ddots & \cdots \\
g_1(y_q) - g_1(y_0) & \cdots & g_q(y_{q}) - g_q(y_0)
\end{pmatrix} 
\end{equation}
has full rank. Then, for all $i=1, \dots, q$ there exist constants $a_{q+1}, \dots, a_n, c\in\mathbb{R}$, such that $f_i(x)=\sum_{j=q+1}^n a_jf_j(x) +c$ for all $x\in\mathcal{X}$. 
\end{lemma}
\begin{proof}
Fix $y_1, y_2\in\mathcal{Y}$, such that $g_1(y_1)\neq g_1(y_2)$. Then, we have for all $x\in\mathcal{X}$
\begin{align*}
&f_1(x)g_1(y_1) + \dots + f_n(x)g_n(y_1) = f(x) + g(y_1),\\&
f_1(x)g_1(y_2) + \dots + f_n(x)g_n(y_2) = f(x) + g(y_2),
\end{align*}
and subtraction of these equalities yields 
$$
f_1(x)[g_1(y_1)- g_1(y_2)] + \dots + f_n(x)[g_n(y_1)-g_n(y_2)] = g(y_1) - g(y_2).
$$
Defining $a_i = g_i(y_1)- g_i(y_2)$ and $c = g(y_1)- g(y_2)$ yields the first result (with $a_1\neq 0$).

Now, we prove the ``Moreover'' part. Consider equalities 
\begin{align*}
f_1(x)g_1(y_0) + &\dots + f_n(x)g_n(y_0) = f(x) + g(y_0),\\
f_1(x)g_1(y_1) + &\dots + f_n(x)g_n(y_1) = f(x) + g(y_1),\\
&\dots\\
f_1(x)g_1(y_q) + &\dots + f_n(x)g_n(y_q) = f(x) + g(y_q),
\end{align*}
where $y_0, \dots, y_q$ are defined, such that matrix (\ref{matrix243}) has full rank. Subtracting from each equality, the first equality yields 
\begin{align*}
f_1(x)[g_1(y_1)- g_1(y_0)] + &\dots + f_n(x)[g_n(y_1)-g_n(y_0)] = g(y_1) - g(y_0)\\
&\dots \\
f_1(x)[g_1(y_q)- g_1(y_0)] + &\dots + f_n(x)[g_n(y_q)-g_n(y_0)] = g(y_q) - g(y_0).
\end{align*}
Using matrix formulation, this can be rewritten as 
\begin{equation}
M\begin{pmatrix}
f_1(x) \\
\cdots \\
f_q(x) 
\end{pmatrix} =
\begin{pmatrix}
g(y_1)-g(y_0) -\sum_{j=q+1}^n f_{j}(x)[g_{j}(y_1) -g_{j}(y_0)] \\
\cdots \\
g(y_q)-g(y_0) -\sum_{j=q+1}^n f_{j}(x)[g_{j}(y_q) -g_{j}(y_0)]
\end{pmatrix} .
\end{equation}
Multiplying both sides by $M^{-1}$ indicates that $f_i(x), i=1, \dots, q$ are nothing else than a linear combination of $f_{q+1}(x), \dots, f_n(x)$, which is what we wanted to show.
\end{proof}

\begin{customthm}{\ref{thmAssymetricMultivariatesufficient}}
Let $(X_1, X_2)$ follow the $CPCM(F_1, \dots, F_k)$ model with graph $X_1\to X_2$, where $F_1, \dots, F_k$ belong to the exponential family of distributions with corresponding sufficient statistics  $T_m= (T_{m,1}, \dots, T_{m,q_m})^\top$, $m=1, \dots, k$.  Following Definition~\ref{CPCM(F1F2)}, let $\tilde{m}\in\{1, \dots, k\}$ be the index such that  $X_2 = F_{\tilde{m}}^{-1}\big(\varepsilon_2; \theta_2(X_1)\big)$. 

The causal graph is identifiable if for all $m\in \{1, \dots, k\}$, at least one of the following holds: 
\begin{itemize}
    \item $ supp(F_m) \neq supp(X_1)$. 
    \item The function \( \theta_2 \) is not a linear combination of the sufficient statistics \( T_{m,1}, \dots, T_{m,q_m} \), i.e., there do not exist coefficients \( a_{i,j}, b_i \in \mathbb{R} \) for \( i = 1, \dots, q_{\tilde{m}} \) and \( j = 1, \dots, q_m \) such that  
   \begin{equation}\tag{\ref{eq158}}
   \theta_{2,i}(x) = \sum_{j=1}^{q_m} a_{i,j} T_{m,j}(x) + b_i, \quad \forall x \in \operatorname{supp}(X_1), \quad \forall i \in \{1, \dots, q_{\tilde{m}}\}.
   \end{equation}  
    \item There do not exist constants \( c_1, \dots, c_{q_m} \in \mathbb{R} \) such that the density of \( X_1 \) satisfies  
   \begin{equation}\tag{\ref{eq007v2}}
   p_{X_1}(x) \propto \frac{h_{m,1}(x)}{h_{\tilde{m},2}[\theta_2(x)]} e^{\sum_{i=1}^{q_m} c_i T_{m,i}(x)}, \quad \forall x \in \operatorname{supp}(X_1),
   \end{equation}  where \( h_{m,1} \) is a base measure associated with \( F_{m} \) and \( h_{\tilde{m},2} \) is the normalizing function of \( F_{\tilde{m}} \), both defined in \hyperref[appendix_exponential_family]{Appendix} \ref{appendix_exponential_family}. 
\end{itemize}

Consequentially, the space of non-identifiable distributions is contained in a $\tilde{d}$-dimensional space, where 
\begin{equation}\tag{\ref{dimension_in_theorem2}}
    \tilde{d} = \sum_{m\in\{1, \dots, k\}:  supp(F_m) = supp(X_1)} (q_m+1)(q_{\tilde{m}}+1) -1 .  
\end{equation}
\end{customthm}

\begin{proof}
\label{Proof of thmAssymetricMultivariatesufficient}{}

If the $CPCM(F_1,\dots, F_k)$ is \textit{not} identifiable, then there exists $m\in\{1, \dots, k\}$ and functions $\theta_1$ and $\theta_2$, such that models 
\begin{equation}
    \label{eq425}
    X_1 = \varepsilon_1, X_2 = F_{\tilde{m}}^{-1}(\varepsilon_2, \theta_2(X_1))\text{, and } X_2 = \varepsilon_2, X_1 = F_m^{-1}(\varepsilon_1, \theta_1(X_2))
\end{equation}generate the same joint density function. For simplifying the notation, let $m=1$ and $\tilde{m}=2$. 

\textbf{1) }Trivially, $X_1$ can not be generated as $X_1 = F_1^{-1}(\varepsilon_1, \theta_1(X_2))$ if $supp(F_1) \neq supp(X_1)$. 

\textbf{2)} For a contradiction, we show that $\theta_2$ is a linear combination of $T_{1,1}, \dots, T_{1,q_m}$. Decompose the joint density as
\begin{equation}\label{eq59}
  p_{(X_1, X_2)}(x,y) = p_{X_1}(x)p_{X_2\mid {X_1}}(y\mid x) = p_{X_2}(y)p_{{X_1}\mid {X_2}}(x\mid y), \,\,\,\,\,\,\,\,x\in supp(X_1), y\in supp(X_2).
 \end{equation}
Since $F_1$ and $F_2$ lie in the exponential family of distributions, we use the notation from \hyperref[appendix_exponential_family]{Appendix} \ref{appendix_exponential_family} and rewrite it as
\begin{equation*}
    \begin{split}
  &     p_{{X_2}\mid {X_1}}(y\mid x) = h_{1,1}(y)h_{1,2}[\theta_2(x)]e^{\sum_{i=1}^{q_2}\theta_{2,i}(x)T_{2,i}(y)},\\&
  p_{{X_1}\mid {X_2}}(x\mid y) = h_{2,1}(x)h_{2,2}[{\theta_1}(y)]e^{\sum_{i=1}^{q_1}{\theta}_{1,i}(y)T_{1,i}(x)}.  
    \end{split}
\end{equation*}
After a logarithmic transformation of both sides of (\ref{eq59}), we obtain 
\begin{equation}\label{eq254}
\begin{split}
\log[p_{(X_1,X_2)}(x,y)] &= \log[p_{X_1}(x)] +  \log[h_{1,1}(y)]+\log\{h_{1,2}[\theta_{2}(x)]\} + \sum_{i=1}^{q_2}\theta_{2,i}(x)T_{2,i}(y) \\&
= \log[p_{X_2}(y)] +  \log[h_{2,1}(x)]+\log\{h_{2,2}[\theta_{1}(y)]\} + \sum_{i=1}^{q_1}\theta_{1,i}(y)T_{1,i}(x).
\end{split}
\end{equation}
Define $f(x) = \log[p_{X_1}(x)] +\log\{h_{1,2}[\theta_{2}(x)]\} -\log[h_{2,1}(x)]$ and $g(y) =\log[h_{1,1}(y)] -  \log[p_{X_2}(y)] + \log\{h_{2,2}[\theta_{1}(y)]\}$. Then, equality (\ref{eq254}) reads as 
\begin{equation}\label{eq9876}
f(x) + g(y) = \sum_{i=1}^{q_1}\theta_{1,i}(y)T_{1,i}(x) - \sum_{i=1}^{q_2}T_{2,i}(y)\theta_{2,i}(x).
\end{equation}
Finally, we use Lemma \ref{PomocnaLemma1}. We know that functions $T_{2,i}$ are affinly independent in the sense presented in Lemma  \ref{PomocnaLemma1} (see (\ref{eq145151}) in \hyperref[appendix_exponential_family]{Appendix} \ref{appendix_exponential_family}). Therefore, Lemma \ref{PomocnaLemma1} gives us that $\theta_{2,i}, i=1, \dots, q_2$ are only a linear combination of $T_{1, j}, j=1, \dots, q_1$, which is what we wanted to show. 

\textbf{3)} For a contradiction, we show that $p_{X_1}$ must have a form \eqref{eq007v2}. Let us rewrite equation~\eqref{eq9876} into 
\begin{equation}\label{eq9876543}
\begin{split}
    \log[p_{X_1}(x)]   =&-\log\{h_{1,2}[\theta_{2}(x)]\} +\log[h_{2,1}(x)]-g(y)\\&  +\sum_{i=1}^{q_1}\theta_{1,i}(y)T_{1,i}(x) - \sum_{i=1}^{q_2}T_{2,i}(y)\theta_{2,i}(x).
\end{split}
\end{equation}
Fix $y\in supp(F_2)$. Using the form of $\theta_{2,i}$ from the previous bullet-point, we can write 
\begin{equation*}
    \begin{split}
      &  \sum_{i=1}^{q_1}\theta_{1,i}(y)T_{1,i}(x) - \sum_{i=1}^{q_2}T_{2,i}(y)\theta_{2,i}(x) = \sum_{i=1}^{q_1}\theta_{1,i}(y)T_{1,i}(x) - \sum_{i=1}^{q_2}T_{2,i}(y)\bigg[ \sum_{j=1}^{q_1}a_{i,j}T_{1,j}(x)+b_i \bigg] \\& = \sum_{i=1}^{q_1}c_iT_{1,i}(x)+d, 
    \end{split}
\end{equation*}
where $c_i = \theta_{1,i}(y) - \sum_{j=1}^{q_2}\sum_{k=1}^{q_1}T_{2,i}(y)a_{i,j}$ and $d =  \sum_{j=1}^{q_2}b_jT_{2,j}(y)$. Therefore, equation~\ref{eq9876543} can be written as 
\begin{equation*}
   \log[p_{X_1}(x)]   =-\log\{h_{1,2}[\theta_{2}(x)]\} +\log[h_{2,1}(x)]  +\sum_{i=1}^{q_1}c_iT_{1,i}(x) +[d-g(y)]  .
\end{equation*}
Applying exponential on both sides, we obtain (\ref{eq007v2}). 

\textbf{Part ''Consequentially'':} We have shown that if \eqref{eq425} holds, then \( \operatorname{supp}(F_m) = \operatorname{supp}(X_1) \), and the joint density \( p_{(X_1, X_2)} \) is uniquely determined by the coefficients \( a_{i,j}, b_i, c_j \in \mathbb{R} \), where \( i = 1, \dots, q_{\tilde{m}} \) and \( j = 1, \dots, q_m \).  

By counting the number of these coefficients, we find that there are \( (q_m+1)(q_{\tilde{m}}+1) -1 \) of them, with the  ``\( -1 \)'' term accounting for the normalization of the density function. Consequently, \eqref{dimension_in_theorem2} follows by summing over all \( m \in \{1, \dots, k\} \).  
\end{proof}

\subsection{Proof of Consequence \ref{consequenceprva}}\label{consequence}
\begin{customconsequence}{\ref{consequenceprva}}
\begin{itemize}
\item Suppose that \( \text{supp}(X_1) = \mathbb{R} \), \(\text{supp}(X_2) = \{0, 1, \dots\}\) such as on Figure~\ref{Asymmetrical_picture}, and let \((X_1, X_2)\) admit the \(CPCM(F_1, F_2)\) model with graph \(X_1 \to X_2\), where \(F_1\) is a Gaussian distribution and \(F_2\) is a Poisson distribution with rate parameter \(\lambda\). The causal graph is identifiable if and only if there do not exist constants \(a_1, a_2, b, c_1, c_2\in \mathbb{R}\), $a_1, c_1<0$, such that for all \(x \in \mathbb{R}\)
\begin{equation*}
\lambda(x) = e^{a_1 x^2 +a_2x + b}, \quad p_{X_1}(x) \propto e^{c_1 x^2 + c_2 x }.
\end{equation*}
\item Let \((X_1, X_2)\) admit the \(CPCM(F)\) model with graph \(X_1 \to X_2\), where \(F\) is a Gamma distribution with parameters \(\theta = (\alpha, \beta)^\top\). If there do not exist constants \(a, b, c, d, e, f \in \mathbb{R}\) such that
\begin{equation*}
\alpha(x) = a\log(x) + bx + c, \quad \beta(x) = d\log(x) + ex + f, \quad \forall x > 0,
\end{equation*}
then the causal graph is identifiable.
\item Let \((X_1, X_2)\) admit the \(CPCM(F_1, F_2)\) model, where \(F_1\) is a Gamma distribution with parameters \(\theta_1 = (\alpha_1, \beta_1)^\top\) and \(F_2\) is a Beta distribution with parameters \(\theta_2 = (\alpha_2, \beta_2)^\top\). If there do not exist constants \(a_i, b_i, c_i, d_i, e_i, f_i \in \mathbb{R}\), \(i = 1, 2\), such that for all \(x \in (0, 1)\)
\begin{equation*}
\begin{split}
\alpha_1(x) &= a_1 \log(x) + b_1 x + c_1, \quad \beta_1(x) = d_1 \log(x) + e_1 x + f_1, \\
\alpha_2(x) &= a_2 \log(x) + b_2 \log(1 - x) + c_2, \quad \beta_2(x) = d_2 \log(x) + e_2 \log(1 - x) + f_2,
\end{split}
\end{equation*}
then the causal graph is identifiable.
\end{itemize}
\end{customconsequence}

\begin{proof}
\label{proof_of_consequence_multi}
Poisson distribution has one parameter and it can be written as $h_1(x) = 1/x!, h_2(x) = e^{-e^x}, T(x) = x$. Note that we do not use classical form of density function but its reparametrisation where $\theta(x) = \log(\lambda(x))$ where $\lambda$ is the classical rate parameter. Theorem \ref{thmAssymetricMultivariatesufficient} gives us that the causal graph is identifiable if there do not exist constants \(a_1, a_2, b, c_1, c_2\in \mathbb{R}\), such that for all \(x \in \mathbb{R}\)
\begin{equation*}
\lambda(x) = e^{a_1 x^2 +a_2x + b}, \quad p_{X_1}(x) \propto e^{c_1 x^2 + c_2 x },
\end{equation*}
then the causal graph is identifiable. It is a simple exercise to prove that the joint distribution is integrable if and only if $a_1, c_1<0$. 

The second and the third bullet-point follow directly from Theorem \ref{thmAssymetricMultivariatesufficient}, noting the following: 
\begin{itemize}
    \item The density function of the \textbf{Gamma distribution} with parameters \(\theta = (\alpha, \beta)^\top\) is given by \(p(x) = \frac{\beta^\alpha}{\Gamma(\alpha)} x^{\alpha - 1} e^{-\beta x}\), \(x > 0\). The sufficient statistics are \([T_1(x), T_2(x)] = [\log(x), x]\).
    \item The density function of the \textbf{Beta distribution} with parameters \(\theta = (\alpha, \beta)^\top\) is given by \(p(x) = \frac{1}{B(\alpha, \beta)} x^{\alpha - 1} (1 - x)^{\beta - 1}\). The sufficient statistics are \([T_1(x), T_2(x)] = [\log(x), \log(1 - x)]\).
\end{itemize}
\end{proof}

%%%%%%%%%%%%%%%%%%%%% Section 2.4%%%%%%%%%%%%%%%%%%%%%%%%%
\subsection{Proof of Lemma \ref{thmMultivairateIdentifiability}}
\begin{customlem}{\ref{thmMultivairateIdentifiability}}
Let $F_{\textbf{X}}$ be generated by the $CPCM(F_1, \dots, F_k)$ with DAG $\mathcal{G}$ and with density $p_\textbf{X}$. Assume that for all $ i,j\in\mathcal{G}$, $ S\subseteq V$, such that $i\in pa_j$ and  $pa_j\setminus \{i\}\subseteq S \subseteq nd_j\setminus\{i,j\}$, there exist $\textbf{x}_{S}{:}\,\,  p_S(\textbf{x}_S)>0$, such that a bivariate model defined as $X=\tilde{\varepsilon}_X, Y = F^{-1}_j\big(\tilde{\varepsilon}_Y, \tilde{\theta}(X)\big)$ is identifiable (in the sense of Definition \ref{DEFidentifiability}), where  $F_{\tilde{\varepsilon}_X} = F_{X_i\mid \textbf{X}_{S} =\textbf{ x}_S}    $ and $\tilde{\theta}(x) = \theta_j(\textbf{x}_{pa_j\setminus\{i\}}, x)$,  $x\in supp(X)$.
Then,  $\mathcal{G}$ is identifiable from the joint distribution. 
 \end{customlem}
 
 \begin{proof}
 \label{Proof of thmMultivairateIdentifiability}
Let there be two $CPCM(F_1, \dots, F_k)$ models, with causal graphs $\mathcal{G}\neq \mathcal{G}'$, that both generate $F_\textbf{X}$.   From Proposition 29 in \cite{Peters2014} (recall that we assume causal minimality of $CPCM(F_1, \dots, F_k)$), there exist variables $L,Y\in \{X_1, \dots, X_d\}$, such that 
\begin{itemize}
\item $Y\to L$ in $\mathcal{G}$ and $L\to Y$ in $\mathcal{G}'$,
\item $S:=\underbrace{\big\{pa_L(\mathcal{G})\setminus\{Y\}\big\}}_\text{\textbf{Q}}\cup\underbrace{\big\{pa_Y(\mathcal{G}')\setminus\{L\}\big\}}_\text{\textbf{R}}\subseteq \big\{nd_L(\mathcal{G}) \cap nd_Y(\mathcal{G}')\setminus\{Y,L\}\big\} $. 
\end{itemize}
For this $S$, select $\textbf{x}_S$ in accordance to the condition in the theorem. Below, we use the notation $\textbf{x}_S=(\textbf{x}_q, \textbf{x}_r)$ where $q\in \textbf{Q}, r\in \textbf{R}$. Now, we use Lemma 36 and Lemma 37 from \citep{Peters2014}.  Since $Y\to L$ in $\mathcal{G}$, we define a bivariate SCM as  \footnote{Informally, we consider  $Y^\star := Y\mid \{\textbf{X}_S=\textbf{x}_S\}$ and $L^\star:=L\mid \{\textbf{X}_S=\textbf{x}_S\}$.} $$Y^\star=\tilde{\varepsilon}_{Y^\star},\,\,\,\,\,\,\,\,\,\, L^\star = F^{-1}_{L}\big(\varepsilon_L; \theta_L(Y^\star)\big),$$ 
where $\tilde{\varepsilon}_{Y^\star} \overset{D}{=} Y\mid \{\textbf{X}_S=\textbf{x}_S\}$ and $\varepsilon_L\indep Y^\star, \varepsilon_L\sim U(0,1)$. This is a bivariate CPCM with $Y^\star\to L^\star$. However, the same holds for the other direction: Since $L\to Y$ in $\mathcal{G}'$, we can also define a bivariate SCM in the following manner: $$L^\star=\tilde{\varepsilon}_{L^\star},\,\,\,\,\,\,\,\,\,\, Y^\star = F^{-1}_{Y}\big(\varepsilon_Y; \theta_Y(L^\star)\big),$$ 
 where $\tilde{\varepsilon}_{L^\star} \overset{D}{=}L\mid \{\textbf{X}_S=\textbf{x}_S\}$ and $\varepsilon_Y\indep L^\star, \varepsilon_Y\sim U(0,1)$. We obtained a bivariate CPCM with $L^\star\to Y^\star$, which is a contradiction with the pairwise identifiability. Hence,  $\mathcal{G}= \mathcal{G}'$.  
 \end{proof}
 
\subsection{Proof of Lemma \ref{lemma_o_overparametrizacii}}
 \begin{customlem}{\ref{lemma_o_overparametrizacii}}
Suppose that the joint distribution $F_{(X_1,X_2)}$ is generated according to the model $CPCM(F_2)$ with graph $X_1\to X_2$, where $F_2$ is a distribution function belonging to the exponential family. 

Then, there exists $F_1$ such that the model $CPCM(F_1)$ with graph $X_2\to X_1$ also generates $F_{(X_1,X_2)}$. In other words, there exists $F_1$ such that the causal graph in $CPCM(F_1, F_2)$ is not identifiable from the joint distribution. 
\end{customlem}
\begin{proof}\label{Proof of lemma_o_overparametrizacii}
The idea of the proof is the following: we select $F_1$, such that its sufficient statistic is equal to $\theta_2$. 

Let us denote the original model as
\begin{equation*}
\begin{split}
 & \,\,\,\,\,\,\,  X_1 = \varepsilon_1, X_2 = F_2^{-1}\big(\varepsilon_2, \theta_2(X_1)\big), \varepsilon_2\sim U(0,1), \varepsilon_1\indep \varepsilon_2,
\end{split}
\end{equation*}
where (using notation from Appendix \ref{appendix_exponential_family}) the conditional density function has a form:
$$
p_{X_2\mid {X_1}}(y\mid x) =h_{2,1}(y)h_{2,2}[\theta_2(x)]\exp[\theta_{2}(x)T_{2}(y)].
$$
We define $F_1$ from an exponential family in the following manner: consider the sufficient statistic $T_1(x) = \theta_2(x)$ for all $x$ in support of $X_1$ and choose $h_{1,1}(x) =  p_{X_1}(x)h_{2,2}[\theta_2(x)]$ and $h_{1,2}(y) = \frac{h_{2,1}(y)}{p_{X_2}(y)}$ for all $y$ in support of $X_2$. Then, a model where 
\begin{equation*}
\begin{split}
 & \,\,\,\,\,\,\,  X_2 = \varepsilon_2, X_1 = F_1^{-1}\big(\varepsilon_1, \theta_1(X_2)\big), \varepsilon_1\sim U(0,1), \varepsilon_1\indep \varepsilon_2,
\end{split}
\end{equation*}
for a specific choice $\theta_1(y) = T_2(y)$ has the following conditional density function: 
$$
p_{X_1\mid {X_2}}(x\mid y) =h_{1,1}(x)h_{1,2}[\theta_1(y)]\exp[\theta_{1}(y)T_{1}(x)] = \frac{p_{X_1}(x)}{p_{X_2}(y)}h_{2,1}(y)h_{2,2}[\theta_2(x)]\exp[\theta_{2}(x)T_2(y)].
$$
Therefore, the joint distribution is equal in both models, since
\begin{equation*}
\begin{split}
   p_{X_1}(x) h_{2,1}(y)h_{2,2}[\theta_2(x)]\exp[\theta_{2}(x)T_{2}(y)] &=p_{X_2}(y) \frac{p_{X_1}(x)}{p_{X_2}(y)}h_{2,1}(y)h_{2,2}[\theta_2(x)]\exp[\theta_{2}(x)T_2(y)]\\
  p_{X_1}(x)p_{X_2\mid {X_1}}(y\mid x) &= p_{X_2}(y)p_{{X_1}\mid {X_2}}(x\mid y).
\end{split}
 \end{equation*}
We found $CPCM(F_1)$ model with graph $X_2\to X_1$ that generates the same distribution. This completes the proof. 
\end{proof}

% ----------------------------------------------------------
% REFERENCES
% ----------------------------------------------------------
%\bibliography{bibliography}

\bibliographystyle{plainnat} 
\bibliography{24-1662.bbl}

\end{document}